\documentclass[12pt]{article}

\usepackage[top=1in, bottom=1.2in, left=1in, right=1in]{geometry} 

\usepackage[ansinew]{inputenc}

\usepackage{makeidx}
\usepackage{textcomp}

\usepackage{amsmath,amsfonts,amssymb,amsthm,amscd,mathrsfs}
\usepackage{bbm}

\usepackage[all]{xypic}
\usepackage{url}
\usepackage{epsfig}
\usepackage{longtable}
\usepackage{subfigure}
\usepackage{algorithm}
\usepackage{algorithmic}

\usepackage{bbm}

\usepackage{setspace}

\usepackage{fancyhdr} 

\newtheorem{thm}{Theorem}[section]
\newtheorem{cor}[thm]{Corollary}
\newtheorem{lem}[thm]{Lemma}

\theoremstyle{definition}
\newtheorem{defn}[thm]{Definition}%[section]

\theoremstyle{definition}

\theoremstyle{remark}

\begin{document}

\title{\textbf{Privacy-Preserving Data Sharing for Genome-Wide Association Studies}}
 \author{Caroline Uhler\thanks{Institute of Science and Technology Austria, Am Campus 1,  3400 Klosterneuburg, Austria,  Email: caroline.uhler@ist.ac.at}, Aleksandra B. Slavkovi\'c\thanks{Department of Statistics, Department of Public Health Sciences, Penn State University, University Park, PA 16802 USA, Email:  sesa@psu.edu }, Stephen E.~Fienberg\thanks{Department of Statistics, Machine Learning Department,  Cylab, and Living Analytics Research Centre, Heinz College, Carnegie Mellon University, Pittsburgh, PA 15213-3890, USA, Email: fienberg@stat.cmu.edu}}
 
 %\author{Caroline Uhler, Aleksandra B. Slavkovi\'c, Stephen E.~Fienberg}

\date{}
\maketitle
\thispagestyle{empty}

\begin{abstract}
%After researchers claimed that, under certain conditions, they could use statistical methods to determine whether an individual with known genotype is part of a mix of DNA samples from which only the minor allele frequencies are known, NIH removed all aggregate GWAS (genome-wide association studies) data from publicly available websites.  While other authors have explained why these claims were exaggerated, the NIH policy has remained in effect because of ongoing concerns that individual privacy could be compromised as a result of publicly-available GWAS databases.

Traditional statistical methods for confidentiality protection of statistical databases do not scale well to deal with GWAS (genome-wide association studies) databases especially in terms of guarantees regarding protection from linkage to external information. The more recent concept of \emph{differential privacy},  introduced by the cryptographic community, is an approach which provides a rigorous definition of privacy with meaningful privacy guarantees in the presence of arbitrary external information, although the guarantees come at a serious price in terms of data utility. %Differential privacy implies that essentially the same conclusions are drawn from a data set whether or not a specific individual's data is part of the data set. 
Building on such notions, we propose new methods to release aggregate GWAS data without compromising an individual's privacy. We present methods for releasing differentially private minor allele frequencies, chi-square statistics and $p$-values. We compare these approaches on simulated data and on a GWAS study of canine hair length involving 685 dogs. We also propose a privacy-preserving method for finding genome-wide associations based on a differentially-private approach to penalized logistic regression.\\ 

\end{abstract}

%\begin{Keywords}
%contingency tables; differential privacy; genome-wide association studies (GWAS); chi-square statistics; logistic regression; $p$-values; single nucleotide polymorphism (SNP).

%\end{IEEEkeywords}

\section{Introduction}
In an article that shocked the genetics community, Homer et al.~\cite{Homer} claimed that, under certain conditions, they could use statistical methods to ``accurately and robustly [resolve]" the presence of an individual with known genotype in a mix of DNA samples from which only the minor allele frequencies (MAFs) are known. Their approach compared the MAFs of a specific individual to the distribution of MAFs in a reference population and the distribution of MAFs in a test population and then used a $t$-test to assess if the individual was part of the test population.

Although proposed specifically for use in a forensic context and only secondarily for breaking privacy, the  Homer et al.~\cite{Homer} ``attack" appeared to be generally applicable.  As a reference population one might use the publicly available single nucleotide polymorphism (SNP) data from the HapMap project\footnote{http://hapmap.ncbi.nlm.nih.gov/} which consists of SNP data from 4 populations varying in size from 45 to 90 individuals. Note that the HapMap data set does not contain any information regarding the health status of the individuals. For the test population one might use  the cases in genome-wide association studies (GWAS), which contain both genotype data and disease status. Before the appearance of the \mbox{article \cite{Homer},} the averaged MAFs of the cases and the averaged MAFs of the controls in a GWAS were typically publicly available. 

In response to Homer et al.~\cite{Homer}, Braun et al.~\cite{Braun} showed that their proposed test  depends heavily on the assumption that the genotypes of the test population, the reference population and the specific person under consideration are samples from the same underlying population and that the SNPs used in the study are independent (i.e., that there is no linkage disequilibrium present). These assumptions are usually not met in practice, and as a consequence, the Homer et al. attack lead to a high false-positive rate, see e.g.~Braun et al.~\cite{Braun}.  Others have criticized Homer et al., suggested alternative formulations of the identification problem, claimed to strengthen the attack or suggested different ways to protect the data, e.g., see~\cite{Clayton,Jacobs,Malin,Lumley,Masca,P3G,Jordan,Visscher,Zhou}.  Despite the apparent limitations of the Homer et al.~attack on the privacy of GWAS participants and the controversial and, we believe, exaggerated nature of their statistical claims, NIH immediately removed from open-access databases all aggregate results such as values of averaged MAFs over cases and controls, chi-square ($\chi^2$)-statistics and $p$-values (see Couzin~\cite{Couzin} and Zerhouni and Nabel~\cite{NIHResponse}). The NIH policy  remains in effect today. Every researcher, who wants to gain access to any of these data sets, needs to go through an elaborate approval process. This is a particularly difficult obstacle for computer scientists, mathematicians or statisticians who do not have a credible record in GWAS research. 

Here we propose methods which allow for the release of aggregate GWAS data without compromising an individual's privacy, and in many ways totally bystep the debate on the validity of the claims by Homer et al.~\cite{Homer} and others on the vulnerability of GWAS databases.  Our GWAS privacy guarantees utilize the concept of \emph{differential privacy}, recently introduced by the cryptographic community (e.g., Dwork et al.~\cite{Dwork}).  Differential privacy provides a rigorous definition of privacy with meaningful privacy guarantees in the presence of arbitrary external information. Our contributions  are as follows:
\begin{itemize}
\item We propose a method for the release of the averaged MAFs for the cases and for the controls in GWAS without compromising an individual's privacy.
\item We compute $\epsilon$-differentially private $\chi^2$-statistics and $p$-values and provide a differentially private algorithm for releasing these statistics for the most relevant SNPs.
\item Conditions such as cancer, heart disease, and diabetes are caused by the interaction of various genes and possibly the environment. Detecting such interaction among SNPs related to a specific phenotype (i.e.,~epistasis) is  a main goal of GWAS. Most methods for finding epistasis are based on a two-stage approach: (1)  Filtering all SNPs, e.g.,  using $\chi^2$-statistics or a simple logistic regression, to reduce the potentially interacting SNPs to a small number; (2) Further examining the loci achieving some threshold  for interactions. For example, Park and Hastie \cite{Park} use a form of  penalized logistic regression to test for detecting gene-gene interactions on a small number of SNPs. By adapting the work of \cite{Bhaskar} and \cite{Chaudhuri} to this methodology, we derive a privacy-preserving method for GWAS, where both stages in the two-stage approach satisfy $\epsilon$-differential privacy.
\end{itemize}

Section 2 describes the basic problem and relevant definitions. In Section 3, we present methods for releasing $\epsilon$-differentially private MAFs, $\chi^2$-statistics and $p$-values, and in {Section 4} we evaluate their statistical utility on data based on a simulation study and on a GWAS study of canine hair length involving 685 dogs. In Section 5, we propose a differentially-private  method for finding genome-wide associations based on a penalized approach to logistic regression.

\section{Main Definitions and Notation}
In a typical GWAS setting, we study the interaction between various SNPs and a binary phenotype, as for 
example the disease status of an individual. The binary phenotype takes values 0 (e.g.,~non-diseased) and 1 (e.g.,~diseased). We denote the total number of individuals in a GWAS by $N$ and assume throughout the paper that the number of cases and controls is equal, i.e.,~there are $N/2$ cases and $N/2$ controls. This corresponds to the usual setting in GWAS and is necessary in order to achieve sufficient power to detect SNPs which are associated with a disease. We denote the total number of SNPs in a GWAS by $M'$ and the number of SNPs for which we would like to release aggregate data by $M$. We assume that the SNPs are polymorphic with only two possible nucleotides.  The SNPs therefore take values 0, 1, and 2 representing the number of minor alleles. We summarize the data for each SNP in a $3\times 2$ contingency table, 
where the count in cell $(i,j)$ consists of the number of individuals with genotype $i$ and disease status $j$. %(e.g., see Table~\ref{tab:3x2}). 
We also assume throughout the paper that all margins of such a $3\times 2$ contingency table are positive. This is motivated by the fact that in GWAS usually all SNPs with a MAF smaller than $0.05$ are removed from the study.  

\begin{defn}
A randomized mechanism $\mathcal{K}$ is $\epsilon$-\emph{differentially private} if, for all data sets $D$ and $D'$ which differ in at most one individual and for any $t\in\mathbb{R}$,
$$\frac{\textrm{Pr}(\mathcal{K}(D)=t)}{\textrm{Pr}(\mathcal{K}(D')=t)}\leq e^{\epsilon}.$$
\end{defn}

\begin{defn}
The \emph{sensitivity} of a function \mbox{$f:\mathcal{D}^N\to\mathbb{R}^d$,} where $\mathcal{D}^N$ denotes the set of all databases with $N$ individuals, is the smallest number $S(f)$ such that 
$$|\!|f(D)-f(D')|\!|_1\leq S(f),$$
for all data sets $D, D'\in \mathcal{D}^N$ differing in a single individual.
\end{defn}

Releasing $f(D)+b$, where $b$ is random noise drawn from a Laplace distribution with mean 0 and scale $\frac{S(f)}{\epsilon}$ satisfies the definition of  $\epsilon$-differential privacy (e.g., see \cite{Dwork}). This type of release mechanism is often referred to as the {\em Laplace mechanism.}

\begin{defn} The KL divergence between two probability distributions $f$ and $g$ is defined by
\begin{equation}
D_{KL}(f||g) = \int^{\infty}_{-\infty}f(x)\,log\frac{f(x)}{g(x)}dx.
\end{equation}
\end{defn}
For the analysis of the simulation results in Section 3 we use the KL divergence to measure the difference between two distributions such as the original $\chi^2$-statistic and
its corresponding $\epsilon$-differentially private version.
%Although the KL divergence is not a perfect metric (unsymmetric),
%it is sufficient to measure the difference between two distributions such as the original $\chi^2$-statistic and
%its corresponding $\epsilon$-differentially private version.

\section{Privacy-Preserving Methodology} 
In this section we compute the sensitivity of MAFs,  $\chi^2$-statistics and $p$-values needed to release the private versions of these statistics for each SNP via the Laplace mechanism. %We derive sampling distributions of the private statistics. %release of their corresponding $\epsilon$-differentially private versions. 
We also describe an $\epsilon$-differentially private algorithm for the release of the latter two quantities for the $M$ most relevant SNPs. 

\subsection{Privacy-Preserving Release of Aggregate MAFs}

We now describe a method for releasing the averaged MAFs for the cases and for the controls in GWAS which satisfies differential privacy. The true data form a table consisting of the MAFs of the cases and the controls for $M$ SNPs; e.g., see Table \ref{table_MAF}. In the following, we compute the amount of Laplace noise we need to add to such a table in order to satisfy $\epsilon$-differential privacy. 

\begin{table}[!b]
\caption{Table showing the averaged MAFs of the cases and the controls for $M$ SNPs.} 
\vspace{0.5cm}
\label{table_MAF} 
\centering 
\begin{tabular}{|c|c|c|c|c|}
\hline \textbf{MAF} & \textbf{SNP 1} & \textbf{SNP 2} & $\cdots$ & \textbf{SNP $M$} \\ \hline
\textbf{Cases} & 0.29 & 0.20 & $\cdots$ & 0.11 \\ \hline
\textbf{Controls} & 0.27 & 0.31 & $\cdots$ & 0.10 \\ \hline
\end{tabular}
\end{table}

\begin{lem}\label{lem_maf}
The sensitivity of the averaged MAFs of the cases and the controls based on $N$ individuals,  with $N/2$ cases and $N/2$ controls, for $M$ SNPs is $\frac{2M}{N}$.
\end{lem}

\begin{proof}
Without loss of generality, we can assume that the individual, whose genotype we can change, belongs to the cases. Denote this individual  by $j$. For a given SNP we denote the number of minor alleles of individual $i$ before adding noise by $a_i$ and the perturbed counts by $a'_i$. Note that $a_i=a'_i$ for all $i\neq j$. In addition, $|a_j-a'_j|\leq 2$. Therefore, for a given SNP we can compute  the sensitivity of the averaged MAF as follows:
$$
\Bigg|\frac{1}{N/2}\sum_{i=1}^{N/2} \frac{a_i}{2} - \frac{1}{N/2}\sum_{i=1}^{N/2} \frac{a'_i}{2}\Bigg| = \frac{1}{N/2}\Bigg| \frac{a_j}{2} -  \frac{a'_j}{2}\Bigg| \leq \frac{2}{N}.
$$
This holds for every SNP. As a consequence, for $M$ SNPs the sensitivity is $\frac{2M}{N}$, namely the 1-norm of the $M$-dimensional vector where all entries are $\frac{2}{N}$.
\end{proof}

%[THIS assumes independence of SNPs, right?? IS THIS AN ISSUE? HOW DOES THIS RELATES TO compositional property? What do we anticipate if we have positive dependence? SHOULD WE DISCUSS THIS?] No, it follows from the definition of sensitivity as the 1-norm of the resulting m-dim vector. 

Lemma~\ref{lem_maf} shows that a data release mechanism that adds Laplace noise with mean 0 and scale $\frac{2M}{N\epsilon}$  to each cell entry in Table \ref{table_MAF} yields $\epsilon$-differential privacy. This result can be seen as a special case of Example 3 in \cite{Dwork} where every cell entry is a histogram by itself. %This is a simple result that can be re-derived via other histogram-based differentially private algorithms (e.g., ) % is a special example of a histogram, where every cell entry is a histogram by itself, and therefore a special case of Example 3 in \cite{Dwork}.

Similarly, if instead of releasing the averaged MAFs, we want to release $M$ $3\times 2$ tables containing the counts for each genotype and disease status, the sensitivity would be $2M$. Therefore, we have to add Laplace noise with mean 0 and scale $\frac{2M}{\epsilon}$ to ensure $\epsilon$-differential privacy.

\subsection{Privacy-Preserving Release of $\chi^2$-Statistics and $p$-Values}

In many GWAS settings, researchers report the $\chi^2$-statistics and the $p$-values of the most relevant SNPs.   We   propose a method for releasing these quantities in a differential privacy-preserving way, by first computing the sensitivity and then modifying a method proposed in \cite{Bhaskar}, for release of frequent itemsets, to release the noisy statistics corresponding to the most relevant SNPs.
 %As for the MAFs in the previous section, we first compute the sensitivity and then modify a method proposed in \cite{Bhaskar} to release the noisy statistics corresponding to the most relevant SNPs.

\begin{thm}
\label{thm_chi}
The sensitivity of the $\chi^2$-statistic based on a $3\times 2$ contingency table with positive margins and $N/2$ cases and $N/2$ controls is $\frac{4N}{N+2}$.
\end{thm}

\begin{proof}
Consider the following  $3\times 2$ contingency table with positive margins and $N/2$ cases and controls each: 
%$$\begin{bmatrix} a & m-a \\ b & n-b \\ N/2-a-b & N/2-m-n+a+b \end{bmatrix}$$
%$$\bordermatrix{ & 0 & 1 \cr 0 & a & m-a \cr 1 & b & n-b \cr 2 & N/2-a-b & N/2-m-n+a+b }$$
%\begin{table}[htb]
%\caption{Cross-classification of individuals by genotype and phenotype for a single SNP}
$$\begin{tabular}{c c | c c |} 
&& \multicolumn{2}{c}{Disease Status}\\ \hline
&& 0 & 1 \\ \hline 
No. Individuals &0 & a & m-a \\ 
With Genotype &1 & b & n-b \\ 
&2 & N/2-a-b & N/2-m-n+a+b \\ \hline
Total && N/2 &N/2\\ \hline
\end{tabular}$$
%\label{tab:3x2}
%\end{table}
\vspace{0.1cm}

\noindent with $a,b\geq 0$, $m,n>0$, $a\leq m$, $b\leq n$, $a+b\leq N/2$, and $m+n<N$.    Let 
\begin{eqnarray*}
\mathcal{D}&=&\{(a,b,m,n)\in \mathbb{N}\mid m>0,\, n>0,\, a\leq m,\, b\leq n,\\ &&\,\,\, a+b\leq N/2,\, m+n<N\}.\end{eqnarray*} Then we can view the $\chi^2$-statistic   as a function
$$\chi^2: \mathcal{D}\longrightarrow\mathbb{R}_{\geq0},$$
where $(a,b,m,n)$ gets mapped to the $\chi^2$-statistic of the corresponding contingency table. The sensitivity corresponds to  the values of $(a,b,m,n)\in\mathcal{D}\cap\{a\geq1\}$, which maximize
$$|\chi^2(a,b,m,n)-\chi^2(a-1,b+1,m-1,n+1)|.$$
Our approach is to compute the sensitivity by maximizing the directional derivative of $\chi^2(a,b,m,n)$ in direction $(-1/2, 1/2, -1/2, 1/2)$. First note that
\begin{eqnarray*}\chi^2(a,b,m,n)&=&\frac{(2a-m)^2}{m}+\frac{(2b-n)^2}{n}\\ &&+\frac{(2a-m+2b-n)^2}{N-m-n}.\end{eqnarray*}
We then compute the directional derivative of $\chi^2(a,b,m,n)$ in direction $(-1/2, 1/2, -1/2, 1/2)$. It is given by
$$\frac{2a^2}{m^2}-\frac{4a}{m}-\frac{2b^2}{n^2}+\frac{4b}{n}.$$
Over $\mathcal{D}\cap\{a\geq1\}$ this is maximized by the smallest possible value of $a$, the largest possible value of $m$, the largest possible value of $b$ and the smallest possible value of $n$. Consequently, the sensitivity is given by:
$$\Bigg|\chi^2\left(\begin{bmatrix} 1 & N/2 \\ N/2-2 & 0 \\ 1 & 0 \end{bmatrix}\right) - \chi^2\left(\begin{bmatrix} 0 & N/2 \\ N/2-1 & 0 \\ 1 & 0 \end{bmatrix}\right) \Bigg|,$$
which we can   easily see to be $\frac{4N}{N+2}$.
\end{proof}

Note that the sensitivity of the $\chi^2$-statistic grows as a function of  $N$, but is asymptotically constant. This is interesting since the $\chi^2$-statistic for a table with fixed frequencies grows proportional to $N$. In order to achieve $\epsilon$-differential privacy for releasing the $\chi^2$-statistic for a single SNP, we need to add Laplace noise with scale $\frac{1}{\epsilon}\frac{4N}{N+2}$ to the true $\chi^2$-statistic. Thus for increasing $N$, the perturbed (private) $\chi^2$-statistics get more accurate.
%%%NEW STUFF %%

Before we consider the sensitivity of the $p$-values, we derive the asymptotic distribution of the perturbed $\chi^2$-statistic which is a convolution of its (asymptotic) sampling distribution and perturbation. 

\begin{thm}
\label{thm:conv}
Let a $\chi^2$ test statistic $T$ have the $\chi^2$ sampling distribution with $2$ degrees of freedom and let the perturbation $Y\sim Laplace(0, 4/\epsilon)$. Then, the distribution of the perturbed $\chi^2$ test statistic, $X=T+Y$, has the following probability density function $$f(x)=\left\{ \begin{array}{ll} 
\frac{\epsilon}{4}\frac{1}{\epsilon+2}\exp\left(\frac{\epsilon x}{4}\right) & \textrm{if $x< 0$}\\ & \\
\frac{\epsilon}{4}\left[\left(\frac{1}{\epsilon-2}+\frac{1}{\epsilon+2}\right)\exp\left(-\frac{x}{2}\right)-\frac{1}{\epsilon-2}\exp\left(-\frac{\epsilon x}{4}\right)\right] & \textrm{if $x\geq0$}
\end{array}, \right.
$$
and the following cumulative distribution function
$$F(x)=\left\{ \begin{array}{ll} 
\frac{1}{\epsilon+2}\exp\left(\frac{\epsilon x}{4}\right) & \textrm{if $x< 0$}\\ & \\
1-\frac{\epsilon}{2}\left(\frac{1}{\epsilon-2}+\frac{1}{\epsilon+2}\right)\exp\left(-\frac{x}{2}\right)+\frac{1}{\epsilon-2}\exp\left(-\frac{\epsilon x}{4}\right) & \textrm{if $x\geq0$}
\end{array}. \right.
$$
%Let $T$ and $Y$ be independent random variables such that $T\sim \chi^2_2$ and $Y\sim Laplace(0, \epsilon/4)$. Then the pdf and cdf of $X=T+Y$, respectively %Denote a $chi^2$ test statistic with $T\sim $ whose sampling distribution is $\sim \chi^2$ with $2$ degrees of freedom. Let $Y\sim Laplace(0, \epsilon/4)$ be a  
\end{thm}

\begin{proof}
Since $T$ and $Y$ are independent random variables, the distribution of $X$ is the convolution of the given $\chi^2$ and $Laplace$ distributions. 
\end{proof}

%This distribution looks very similar to the underlying Laplace distribution (see Results Section), and the researcher can now compute the $p$-values for the test of independence using the perturbed $chi^2$ statistics. 
We show through simulations in Section \ref{results} that the finite sample distribution is well-approximated by this asymptotic distribution even for tables with low total count, marginal counts or individual counts. This is in contrast to the poor finite sample behavior of the $\chi^2$ test statistics arising when the noise is added directly to the underlying cell counts (see Section \ref{results}); the latter mechanism has been considered by many (e.g., \cite{Dwork, Fienberg_2010}).  For related simulations that demonstrate the interactive effect of sample size and privacy level $\epsilon$ and compare asymptotic efficiency of private and non-private estimators for $2\times 2$ tables and the corresponding $\chi^2$-statistics, see \cite{vu2009differential}. 

We now prove that the asymptotic distribution of the perturbed $\chi^2$-statistic arising from perturbing the cell counts is the same as for the unperturbed $\chi^2$-statistic, namely a $\chi^2$-distribution with two degrees of freedom.
\begin{thm}
\label{thm:cellchi}
Let $X^{(n)}$ denote a 6-dimensional random variable corresponding to the entries of a $3\times 2$ contingency table based on $n$ individuals. Let $Y$ denote a 6-dimensional random variable drawn from $\textrm{Laplace}(0,\frac{2}{\epsilon})$. Then the perturbed $\chi^2$-statistic arising from perturbed cell counts $(X^{(n)}+Y)$ asymptotically has a $\chi^2$-distribution with two degrees of freedom.
\end{thm}
%Let $X^{(n)}$ denote a 6-dimensional random variable corresponding to the entries of a contingency table based on $n$ individuals.

\begin{proof}
Let $p_0, p_1, p_2, q_0, q_1\in [0,1]$ such that $p_0+p_1+p_2=1$ and $q_0+q_1=1$. Under the null hypothesis of independence on a $3\times 2$ contingency table the data is sampled from a multinomial distribution with probability vector $\hat{p}=(p_0q_0, p_0q_1, p_1q_0, p_1q_1, p_2q_0, p_2q_1)^T$. The central limit theorem implies that
$$\sqrt{n}\left(\frac{X^{(n)}}{n}-\hat{p}\right) \xrightarrow{d}\mathcal{N}\left(0, \Sigma\right),$$
where $\Sigma$ is the covariance matrix of the product multinomial, i.e. 
$$\Sigma=\Gamma-\hat{p}\hat{p}^T$$
and $\Gamma=\textrm{diag}(\hat{p})$. Note that $\Sigma$ has rank 2 and therefore also $\Gamma^{-\frac{1}{2}}\Sigma\Gamma^{-\frac{1}{2}}$. Let $Y\sim\textrm{Laplace}(0,\frac{2}{\epsilon})$. Slutsky's theorem
 implies that
 $$\sqrt{n}\left(\frac{X^{(n)}+Y}{n}-\hat{p}\right) \xrightarrow{d}\mathcal{N}\left(0, \Sigma\right),$$
 and therefore that
  $$\sqrt{n}\,\,\Gamma^{-\frac{1}{2}}\left(\frac{X^{(n)}+Y}{n}-\hat{p}\right) \xrightarrow{d}\mathcal{N}\left(0, \Gamma^{-\frac{1}{2}}\Sigma\Gamma^{-\frac{1}{2}}\right).$$
Finally, by invoking the continuous mapping theorem, we  prove the claim, namely
$$\chi^2_{\textrm{perturbed}}=n\left(\frac{X^{(n)}+Y}{n}-\hat{p}\right)^T\Gamma^{-1}\left(\frac{X^{(n)}+Y}{n}-\hat{p}\right)\xrightarrow{d}\chi^2_2.$$
\end{proof}

%%%%%%%end new stuff for comment (b)%%%

%This distribution looks very similar to the underlying Laplace distribution (see Results Section), 
Given the above derived distributions, the researcher can now compute the $p$-values for the test of independence using the perturbed $\chi^2$-statistics (when perturbing the test statistic itself or when adding noise at the level of the cell counts). 

We also consider releasing differentially private $p$-values (without perturbing the counts or the related statistic first). We perform a similar sensitivity analysis on the $p$-values corresponding to the $\chi^2$-statistics when assuming a $\chi^2$-distribution with $2$ degrees of freedom as null distribution, cf. \cite{BFH}. 

\begin{thm}
\label{thm_p}
The sensitivity of the $p$-values of the $\chi^2$-statistic for a $3\times 2$ contingency table with positive margins and $N/2$ cases and  $N/2$ controls is $\exp(-2/3)$, when the null distribution is a $\chi^2$-distribution with $2$ degrees of freedom.
\end{thm}

\begin{proof}
Under the null $\chi^2$-distribution with $2$ degrees of freedom, the $p$-value corresponding to a value $x$ of the $\chi^2$-statistic is
$$\exp(-\frac{x}{2}), \qquad x\geq 0.$$
The first derivative in absolute value is maximized by \mbox{$x=0$.} Therefore, the sensitivity of the $p$-value is given by a change of 1 unit  in a contingency table with $\chi^2=0$, i.e., in a contingency table of the form
$$\begin{bmatrix} a & a \\ b & b \\ N/2-a-b & N/2-a-b \end{bmatrix},$$
where $a,b>0$, and $a+b<N/2$. We therefore need to find $a,b$ which maximize
\begin{equation*}\begin{split}\Bigg|\textrm{$p$-value}\left(\begin{bmatrix} a & a \\ b & b \\ N/2-a-b & N/2-a-b \end{bmatrix}\right) - \qquad\quad\quad\quad\\ \qquad\quad\quad\quad \textrm{$p$-value}\left(\begin{bmatrix} a-1 & a \\ b+1 & b \\ N/2-a-b & N/2-a-b \end{bmatrix}\right) \Bigg|,\end{split}\end{equation*}
where $a,b>0$, and $a+b<N/2$. Equivalently, we need to maximize 
$$\chi^2\left(\begin{bmatrix} a-1 & a \\ b+1 & b \\ N/2-a-b & N/2-a-b \end{bmatrix}\right)$$
over $a,b>0$, and $a+b<N/2$. The corresponding $\chi^2$-statistic is given by
$$\frac{1}{2a-1}+\frac{1}{2b+1},$$
which is maximized by $a=b=1$ and results in a $\chi^2$-statistic of 4/3. Consequently, the sensitivity of the $p$-value is   $\exp(-2/3)$.
\end{proof}

The $\epsilon$-differentially private mechanism for a single SNP would then release a private $p$-value equal to the original value plus Laplace noise with mean zero and scale $\frac{1}{\epsilon}\exp(-2/3)$.

The sensitivity of the $\chi^2$-statistic corresponds to the most `dependent' contingency table while the sensitivity of the $p$-value is determined by an `independent' contingency table. By the most `dependent' (resp.~`independent') contingency table we mean a table which achieves the maximal (resp.~minimal) $\chi^2$-statistic over all contingency tables with $N$ individuals. The maximal $\chi^2$-statistic is $N$, while the minimal $\chi^2$-statistic is $0$. 

Since in practice we are not interested in contingency tables with very large $p$-values, we in effect have overestimated the sensitivity of the $p$-value, and wish instead to determine the sensitivity of the $p$-value within the range of ``interesting" contingency tables.  We therefore analyze what happens if we project all $p$-values, which are larger than a given value $p^*$, onto $p^*$. Since the $\chi^2$-statistic for a table with fixed marginal frequencies grows in proportion to $N$, we analyze the situation where $p^*$ decreases with increasing $N$, i.e.,~$p^*=\exp(-N/c)$, where $c$ is some constant to be specified by the user. Such a $p$-value corresponds to a table with $\chi^2$-statistic $2N/c$ and can be viewed as a contingency table which is at least $N/c$ steps of Hamming distance 1 away from independence.

\begin{cor}
\label{cor}
Projecting all $p$-values which are larger than $p^*=\exp(-N/c)$ onto $p^*$ results in a sensitivity of 
$$\exp\left(-\frac{N}{c}\right) - \exp\left(-\frac{N(2Nc-4N-4c+c^2)}{2c(Nc-2N-c)}\right)$$ for any fixed constant $c\geq 3$, which is a factor of $N/2$.
\end{cor}

\begin{proof}
The proof is similar to the proofs of Theorem \ref{thm_chi} and Theorem \ref{thm_p}. We here give an overview. The contingency table
$$\begin{bmatrix} 0 & \frac{N}{c} \\ \frac{N}{c} & 0 \\ \frac{N(c-2)}{2c} & \frac{N(c-2)}{2c} \end{bmatrix}$$
has a $\chi^2$-statistic $\frac{2N}{c}$ and hence a $p$-value of $\exp(-N/c)$. This table has the maximal $\chi^2$-statistic over all tables which are $N/c$ steps of Hamming distance 1 away from independence, i.e.,~this table is $N/c$ steps away from the following table
$$\begin{bmatrix} \frac{N}{2c} & \frac{N}{2c} \\ \frac{N}{2c} & \frac{N}{2c} \\ \frac{N(c-2)}{2c} & \frac{N(c-2)}{2c} \end{bmatrix}.$$
The largest change in $\chi^2$-statistic is achieved by moving one individual from cell $(3,2)$ to cell $(1,2)$ resulting in the table
$$\begin{bmatrix} 0 & \frac{N+c}{c} \\ \frac{N}{c} & 0 \\ \frac{N(c-2)}{2c} & \frac{N(c-2)-2c}{2c} \end{bmatrix}.$$
This new contingency table has $\chi^2$-statistic 
$$\frac{N(2Nc-4N-4c+c^2)}{c(Nc-2N-c)}.$$
\end{proof}

For large $N$, 
$$\frac{N(2Nc-4N-4c+c^2)}{c(Nc-2N-c)}\approx \frac{2N}{c},$$
and  the corresponding $p$-value is of the order \mbox{of $p^*$.}

In GWAS settings, however, researchers typically  provide only the $\chi^2$-statistics or the corresponding $p$-values of {\em the $M$ most significant SNPs}. Since the ranking reveals additional information, it is not sufficient to add the above computed noise to these statistics in order to achieve differential privacy. Bhaskar et al.~\cite{Bhaskar} show in the context of frequent pattern recognition how to release the most significant patterns together with their frequencies while satisfying differential privacy. %In the following, 
We adapt their method by incorporating our results from Theorem \ref{thm_chi} and Theorem \ref{thm_p} to GWAS, and state the main result of this section: Algorithm 1 for releasing the private $\chi^2$-statistics (p-values) of the $M$ most relevant SNPs. 

Let $M'$ denote the total number of SNPs in a GWAS and $M$ the number of statistics one would like to release. Naively, one might expect that it is necessary to add Laplace noise with scale $\frac{M'}{\epsilon}\frac{4N}{N+2}$ for the $\chi^2$-statistics and $\frac{M'}{\epsilon}\exp(-2/3)$ for the $p$-values. As we see in Algorithm 1, however, the Laplace noise only scales with the number of actually released statistics $M$.

\begin{algorithm}[!h]
%\begin{alg}
\caption{$\epsilon$-Differentially Private Algorithm for Releasing the $M$ Most Relevant SNPs}
\begin{algorithmic}
\begin{STATE}

{\bf Input:}
The $\chi^2$-statistics (resp.~$p$-values) for all $M'$ SNPs and the number of statistics, $M$, we want to release. 

{\bf Output:} 
%The output consists of 
The $M$ noisy $\chi^2$-statistics (resp.~$p$-values).

\begin{enumerate}
\item[1.] Add Laplace noise with mean zero and scale $\frac{4M}{\epsilon}\frac{4N}{N+2}$ to the $\chi^2$-statistics (resp.~Laplace noise with mean zero and scale $\frac{4M}{\epsilon}\exp(-2/3)$ to the $p$-values).
\item[2.] Pick the top $M$ SNPs with respect to the perturbed $\chi^2$-statistics (resp.~$p$-values). We denote the corresponding set of SNPs by $\mathcal{S}$.
\item[3.] Add new Laplace noise with mean zero and scale $\frac{2M}{\epsilon}\frac{4N}{N+2}$ to the true $\chi^2$-statistics of the SNPs in $\mathcal{S}$ (resp.~Laplace noise with mean zero and scale $\frac{2M}{\epsilon}\exp(-2/3)$ to the true $p$-values) and release these perturbed statistics.
\end{enumerate}
%\end{alg}

\end{STATE}
\end{algorithmic}
\label{alg_chi_p}
\end{algorithm}

\begin{thm}
Algorithm \ref{alg_chi_p} is $\epsilon$-differentially private. 
\end{thm}

\begin{proof}
Using the sensitivities computed in Theorem \ref{thm_chi} and Theorem \ref{thm_p}, the proof follows immediately from Theorem 5 in \cite{Bhaskar}.
\end{proof}

\section{Evaluation of Methodology and Results}
\label{results}

We now evaluate the performance of the proposed methods based on data from a simulation study and using a GWAS data set consisting of 685 dogs and their hair length. The GWAS data for the hair length of dogs has first been presented and studied in \cite{CNQ+09} and further been analyzed in \cite{our_paper_GWAS}. It consists of $685$ dogs, $319$ dogs with long hair as cases and $364$ with short hair as controls, and contains $40,842$ SNPs. Cadieu et al.~\cite{CNQ+09} have shown that the long versus short hair phenotype is associated with a mutation in the \emph{fibroblast growth factor-5} (\emph{FGF5} gene) and the largest $\chi^2$-statistic is achieved by a SNP located on chromosome 32 at position $7,100,913$, i.e.,~about $300$Kb apart from \emph{FGF5}.

We also use the simulations from \cite{our_paper_GWAS} performed using HAP-SAMPLE \cite{HAP-SAMPLE}.  HAP-SAMPLE generates the cases and controls by resampling from HapMap. The simulated data show linkage disequilibrium and allele frequencies similar to real data. The simulated association studies consist of 400 cases and 400 controls with about 10,000 SNPs per individual (SNPs typed with the Affy CHIP on chromosome 9 and chromosome 13 of the Phase I/II HapMap data). Two SNPs were chosen to be causative and the simulations were performed for three different MAFs (0.1, 0.25 and 0.4) and two different models of interaction (additive effect and multiplicative effect of the two SNPs). See \cite{our_paper_GWAS} for more details.

For this paper, we omit the simulation results on the statistical utility of $\epsilon$-differentially private release of aggregate MAFs.  Our results are similar to those reported in the current literature on Laplace mechanism for noise addition to histograms or smaller contingency tables with proportions (e.g., \cite{Dwork}, \cite{vu2009differential}). Instead, we focus on the release of differentially-private $\chi^2$-statistics, $p$-values and the most relevant SNPs. 

\subsection{Asymptotic distribution of the perturbed $\chi^2$-statistic}

%%NEW response to comment (b) and (c)
We first present results on the asymptotic distribution of the perturbed $\chi^2$-statistic arising from adding noise directly to the statistic, as derived in Theorem~\ref{thm:conv}, and evaluate the accuracy of the asymptotic approximation. The distribution for $\epsilon=0.2$ is described in Figure \ref{fig:convolution},  and a comparison of three distributions, namely the asymptotic $\chi^2$-distribution, the asymptotic Laplace distribution and their convolution for different values of the privacy parameter $\epsilon$ are shown in Figure  \ref{fig:comparison_distributions}; we can observe that the asymptotic distribution of the perturbed $\chi^2$-statistic is very similar to the underlying Laplace distribution as expected based on the convolution derived in Theorem~\ref{thm:conv}. 

\begin{figure}[!t]
\centering
\includegraphics[scale=0.52]{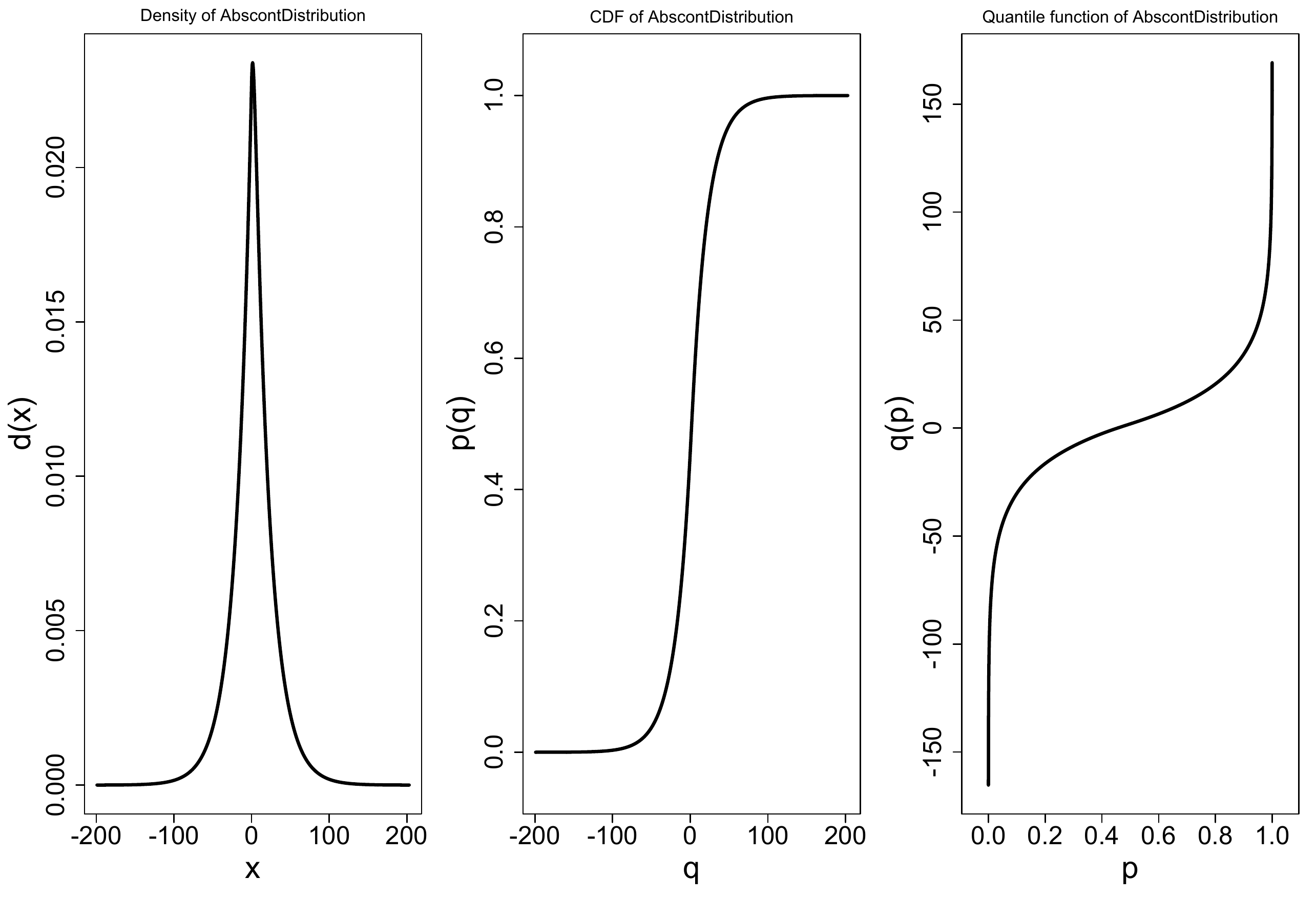} 
\caption{Asymptotic distribution of the perturbed $\chi^2$ test statistic for $\epsilon=0.2$: density function (left), cumulative distribution function (middle), and quantile function (right).}
\label{fig:convolution}
\end{figure}

\vspace{0.4cm}

\begin{figure}[!b]
\centering
\includegraphics[scale=0.52]{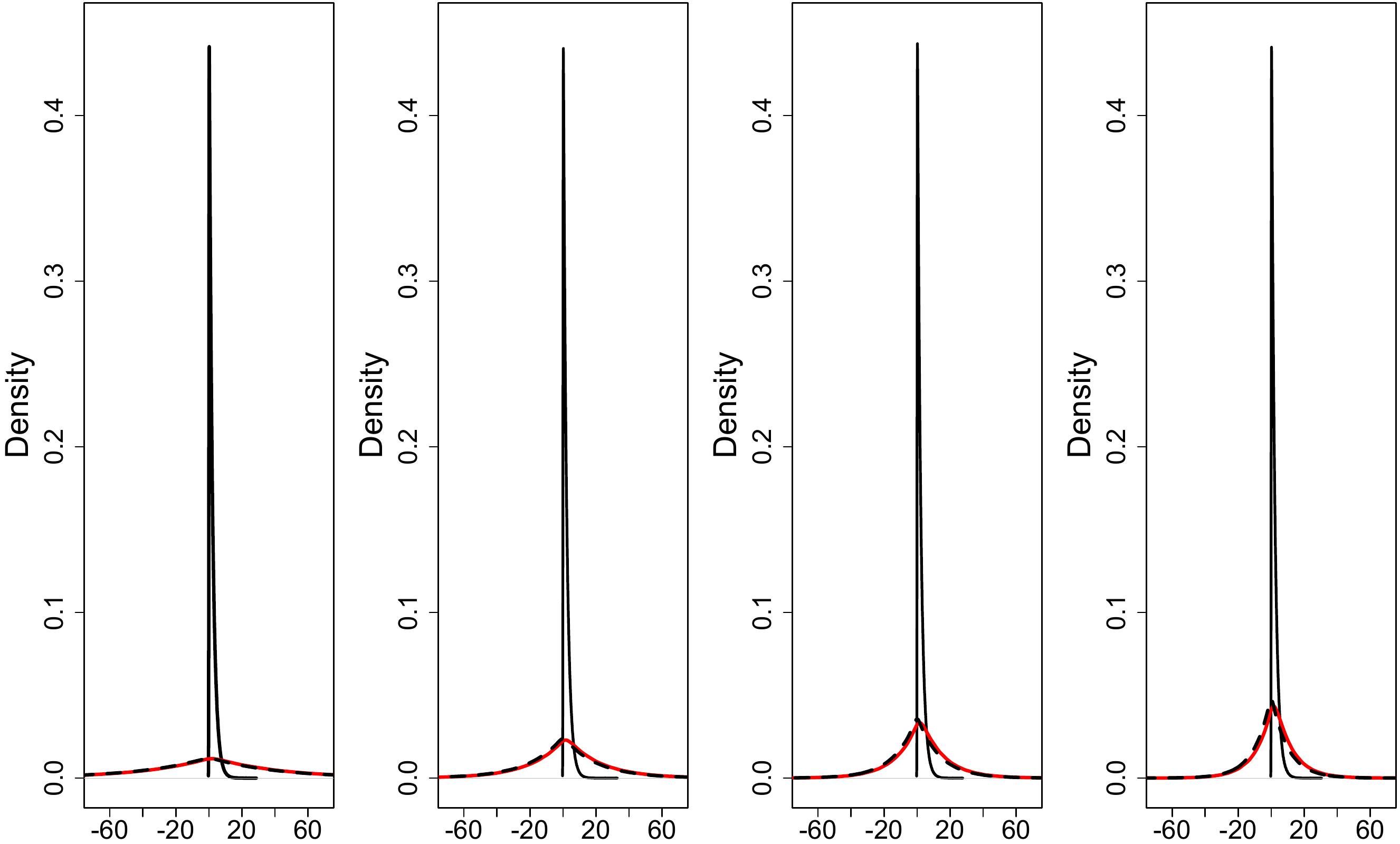} 
\caption{Comparison of the asymptotic sampling distribution (black line), perturbation (black dotted line) and its convolution (red line) for $\epsilon=0.1$ (left), $\epsilon=0.2$ (middle left), $\epsilon=0.3$ (middle right), and $\epsilon=0.4$ (right).}
\label{fig:comparison_distributions}
\end{figure}

Through simulations we analyzed at which point the asymptotic approximation seems to be accurate for finite samples. It turns out that even for tables with very small cell counts or marginal counts, the finite sample distribution of the private $\chi^2$-statistic is well-approximated by its asymptotic distribution, although it is well known that the exact distribution of the original $\chi^2$-statistic is very poorly approximated by the $\chi^2$-distribution for small samples. As an example we discuss the following table:
$$\begin{bmatrix} 1 & 3 \\ 8 & 12 \\ 41 & 35 \end{bmatrix}.$$
We ran a Markov chain on the set of contingency tables which have the same margins as the above table using tools from Algebraic Statistics, namely elements of a Markov basis as moves (e.g., see \cite{Diaconis}). At each step (table), we computed the corresponding $\chi^2$-statistic and added Laplace noise with scale $\frac{4}{\epsilon}$. The resulting posterior distribution is an approximation to the true distribution of the perturbed $\chi^2$-statistic and corresponds to the black dotted line in Figure \ref{fig:pert_chisq1}. The asymptotic distribution of the perturbed $\chi^2$-statistic derived in Theorem ~\ref{thm:conv} is shown in red. These plots  and additional simulations show that the asymptotic approximation is accurate even for tables with a low total count, marginal counts or individual cell counts.

\begin{figure}[!b]
\centering
\includegraphics[scale=0.55]{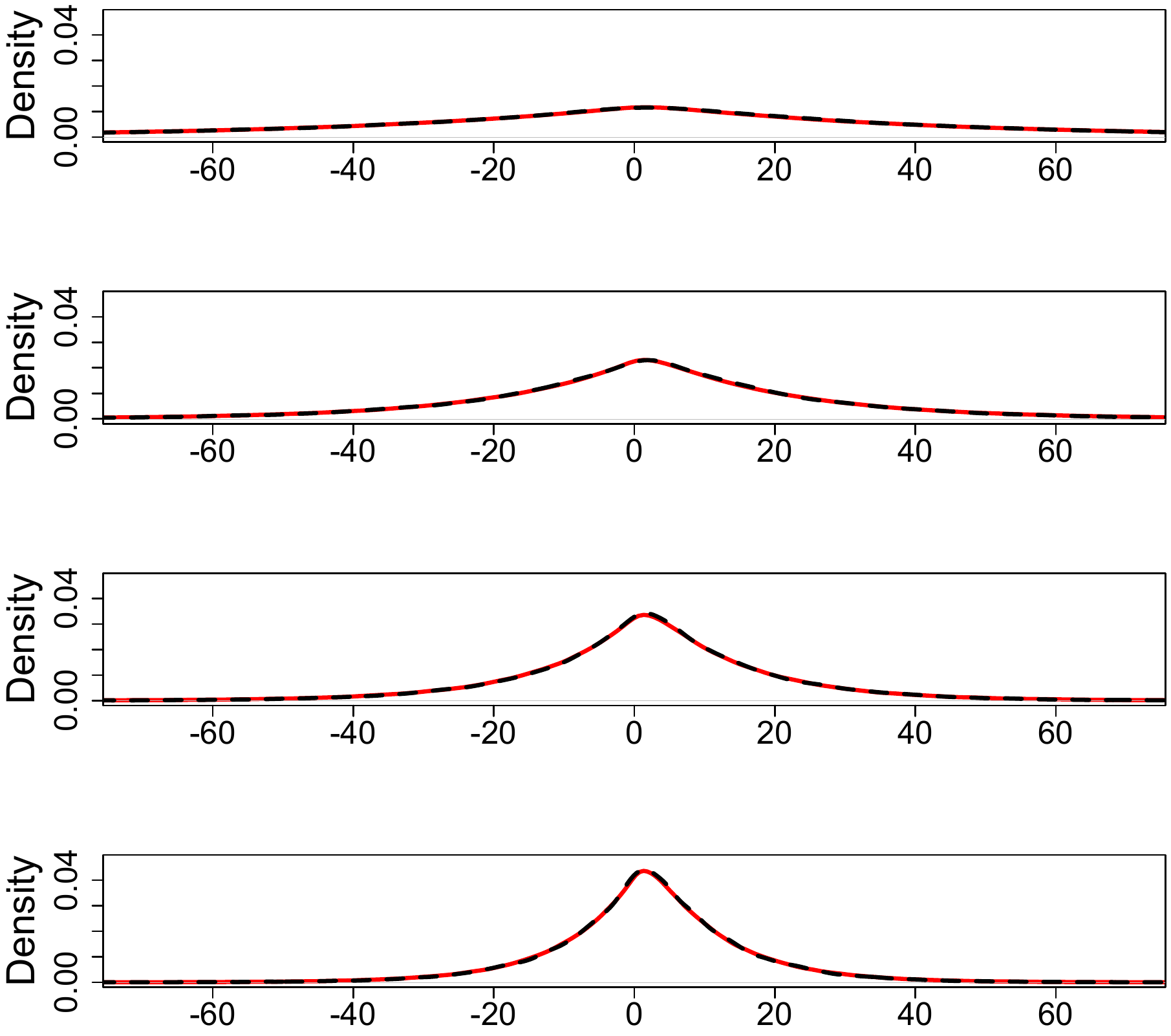} 
\caption{Asymptotic distribution of the perturbed $\chi^2$-statistic (red line) and its true distribution (black dotted line).}
\label{fig:pert_chisq1}
\end{figure}

Similarly, we now analyze under which conditions the asymptotic distribution of the perturbed $\chi^2$-statistic arising from perturbing the cell counts, as shown in Theorem~\ref{thm:cellchi}, appears to be accurate for finite samples. As we will see, when adding noise to the cell counts instead of the $\chi^2$-statistic, the asymptotic distribution of the computed statistic is only accurate for a very large total cell count. We analyze the following tables, one with a total cell count of 10,000 and two with a total cell count of 100,000:
$$(1)\;\;\begin{bmatrix} 1400 & 1600 \\ 1900 & 1300 \\ 1700 & 2100 \end{bmatrix}, \qquad (2)\;\;\begin{bmatrix} 14000 & 16000 \\ 19000 & 13000 \\ 17000 & 21000 \end{bmatrix}, \qquad (3)\;\;\begin{bmatrix} 1 & 3 \\ 26000 & 21000 \\ 23999 & 28997 \end{bmatrix}.$$

We again ran a Markov chain on the set of contingency tables which have the same margins as the above tables using a Markov basis to move between tables. At each step we perturbed the counts by adding Laplace noise with scale $\frac{2}{\epsilon}$ and computed the corresponding perturbed $\chi^2$-statistic. The resulting posterior distribution is an approximation to the true distribution of the perturbed $\chi^2$-statistic and is shown in Figure \ref{fig:pert_cells} for four values of the privacy parameter $\epsilon$. Also the true distribution of the unperturbed $\chi^2$-statistic and the $\chi^2$-distribution are shown for comparison. Note that a total cell count of $10,000$ is not sufficient for a good approximation of the finite sample distribution by the asymptotic distribution. For a total cell count of $100,000$ the approximation appears to be accurate as long as the individual cell counts and margins are not too small, as in the case of the third table.

\begin{figure}[!b]
\centering
\subfigure[Table (1)]{\includegraphics[scale=0.28]{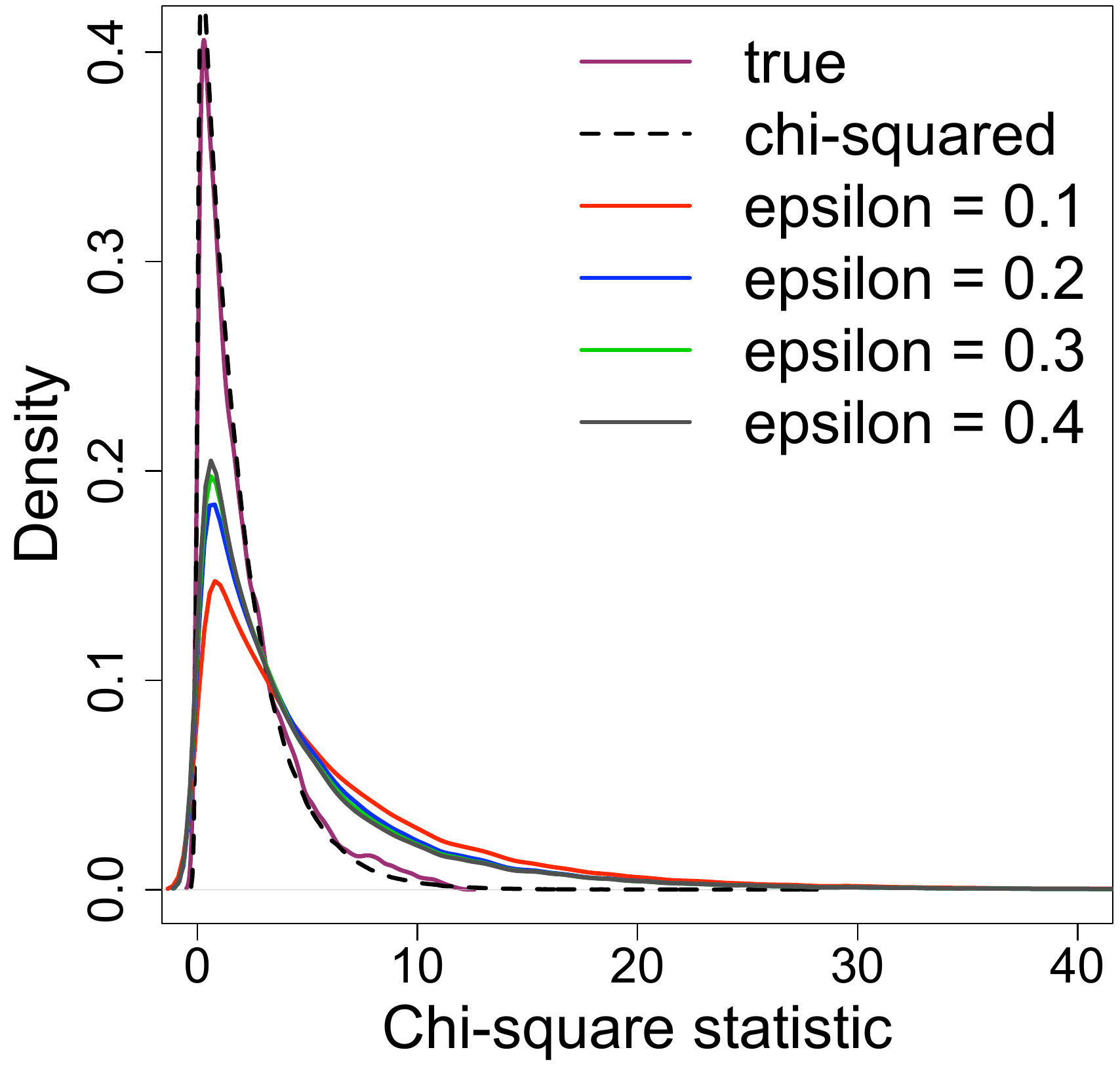}\label{fig:pert_cells:1}} \;
\subfigure[Table (2)]{\includegraphics[scale=0.28]{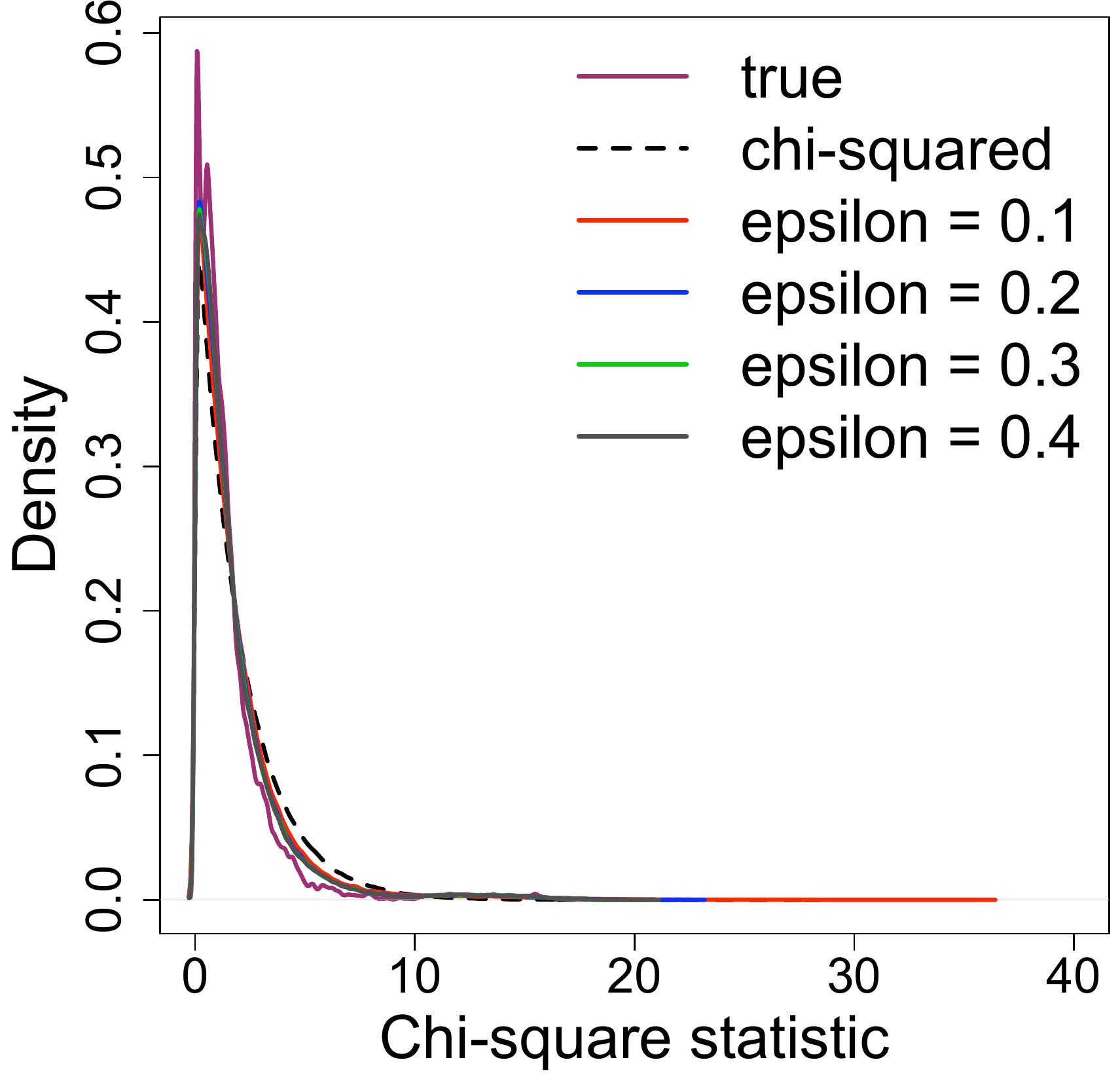}\label{fig:pert_cells:2}} \;
\subfigure[Table (3)]{\includegraphics[scale=0.28]{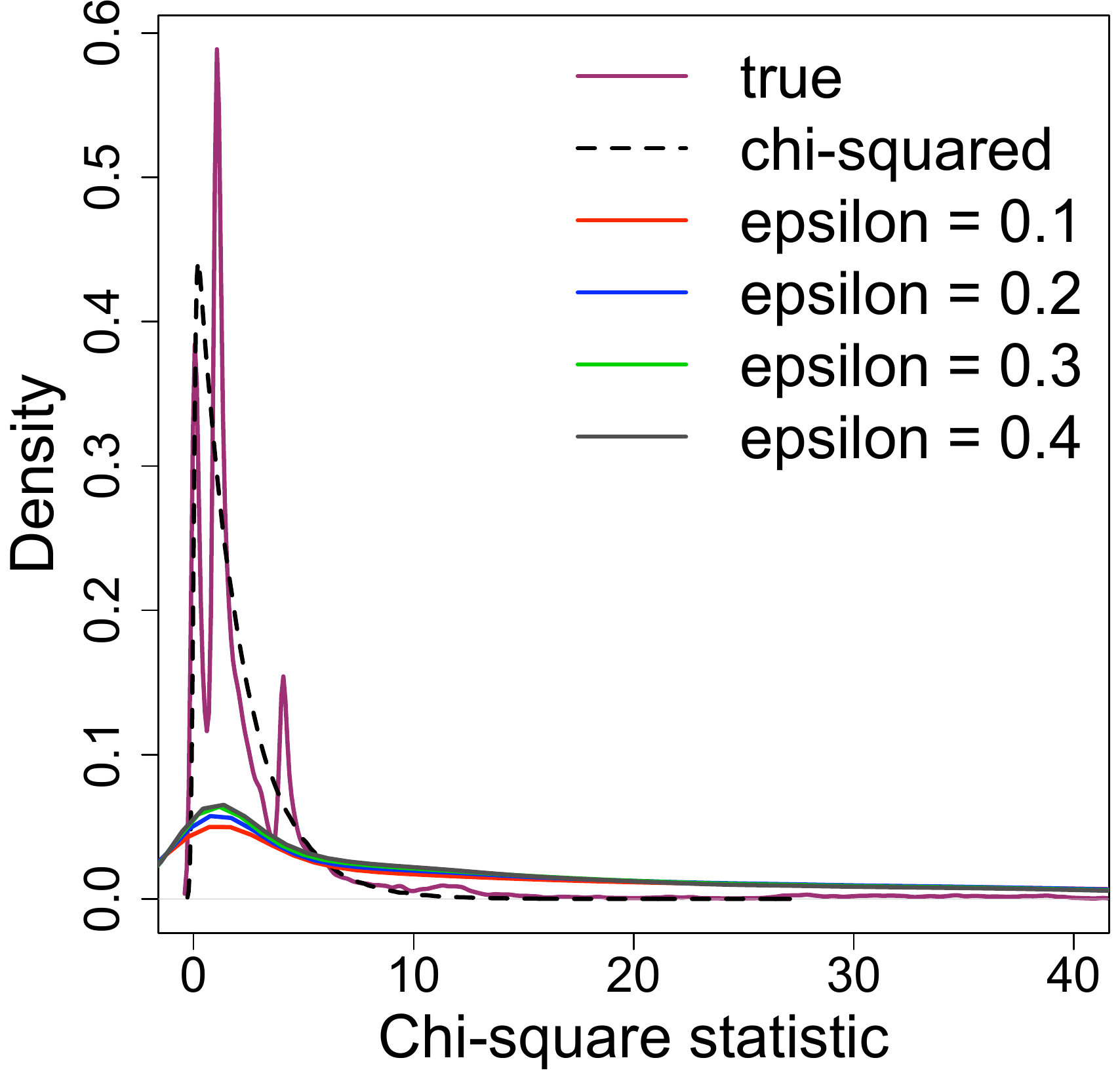}\label{fig:pert_cells:3}}
\caption{Exact and asymptotic distribution of the perturbed and unperturbed $\chi^2$-statistic.}
\label{fig:pert_cells}
\end{figure}

%%%end of NEW response to comment (b). 

\subsection{Differentially-Private $\chi^2$-Statistics}
Based on the results from the previous section, releasing differentially private $\chi^2$-statistics versus perturbing cell counts and then computing the perturbed statistics, seems to work better on finite (and smaller) samples. Next, we focus on evaluating the statistical utility of the proposed release mechanism following Theorem~\ref{thm_chi}.  We compare the $\epsilon$-differentially private $\chi^2$-statistic to the original statistic via KL divergence. We generated $3\times 2$ contingency tables with positive margins and $N/2$ cases and  $N/2$ controls  assuming a product-multinomial distribution with the following frequencies:

\begin{equation}
\label{eq_frequencies}
\begin{split}
(a) \begin{bmatrix} 0.72 & 0.20 \\ 0.18 & 0.28 \\ 0.10 &0.52 \end{bmatrix}, \qquad (b) \begin{bmatrix} 0.60 & 0.23 \\ 0.21 & 0.30 \\ 0.19 & 0.47 \end{bmatrix},\\ (c) \begin{bmatrix} 0.47 & 0.25 \\ 0.45 & 0.51 \\ 0.08 &0.24 \end{bmatrix}, \qquad (d) \begin{bmatrix} 0.65 & 0.46 \\ 0.29 & 0.43 \\ 0.06 & 0.11 \end{bmatrix}.\end{split}
\end{equation}
For the $\chi^2$-distribution with 2 degrees of freedom, an observed  value of 6 corresponds to a  $p$-value of $\exp(-3)\approx 0.05$. The preceding frequency tables correspond to contingency tables for which we expect a $p$-value of 0.05 for
$$(a)\, N=20, \quad (b)\, N=40, \quad (c)\, N=80, \quad (d)\, N=160.$$
For example, for $N=200$ individuals and underlying frequency table (a) we expect a  table of the form
$$\begin{bmatrix} 72 & 20 \\ 18 & 28 \\ 10 &52 \end{bmatrix},$$
which has a $\chi^2$-statistic of 60. Therefore, for $N=20$ we expect a $\chi^2$-statistic of 6. If we fix the number of individuals $N$, then the $\chi^2$-statistic corresponding to frequency table (a) is the largest, namely 8 times the $\chi^2$-statistic corresponding to frequency table (d).

The choice of the frequency tables in (\ref{eq_frequencies}) is motivated by the GWAS on the hair length of dogs in \cite{CNQ+09} and our simulations using HAP-SAMPLE. The $\chi^2$-statistic resulting from the frequency table (a) is comparable to the $\chi^2$-statistic of the SNP most associated to the hair length in dogs (on chromosome 32 at position $7,100,913$ in the CanMap data set). The $\chi^2$-statistic resulting from the frequency table (c) is comparable to the $\chi^2$-statistic of a causative SNP in a simulated association study under the additive model  (i.e.,~main effects only model) for $MAF = 0.4$, and (d) is comparable to a causative SNP under the additive model for $MAF = 0.25$. The frequency table (b) corresponds to an intermediate model for a causative SNP with high MAF and was added for consistency.

For a fixed total number of individuals $N$, we generated 10,000 tables from the frequency tables in (\ref{eq_frequencies}) and computed the corresponding $\chi^2$-statistics. We also generated 10,000 private $\chi^2$-statistics according to the Laplace mechanism described following Theorem~\ref{thm_chi}. In Figure \ref{fig_chi} we plotted the KL divergence between the original and the private $\chi^2$-statistics for increasing $N$ and for four different levels of privacy. The four plots correspond to the four frequency tables in (\ref{eq_frequencies}). We see that the KL divergence depends on the $\chi^2$-statistic of the underlying frequency table, the total number of individuals $N$, and the privacy level $\epsilon$.  Since the added noise is asymptotically $Laplace(0,4)$ distributed, the larger the original $\chi^2$-statistic, the smaller the KL divergence is. Similarly, a larger number of individuals $N$ leads to a larger $\chi^2$-statistic and hence to a smaller KL divergence. The scale of the Laplace noise is inverse proportional to the privacy parameter $\epsilon$. Therefore, the smaller $\epsilon$, the larger the KL divergence is. These simulations demonstrate that it is possible to release $\epsilon$-differentially private $\chi^2$-statistics and maintain good statistical utility in a realistic GWAS setting.

\begin{figure}[!t]
\centering
\includegraphics[scale=0.23]{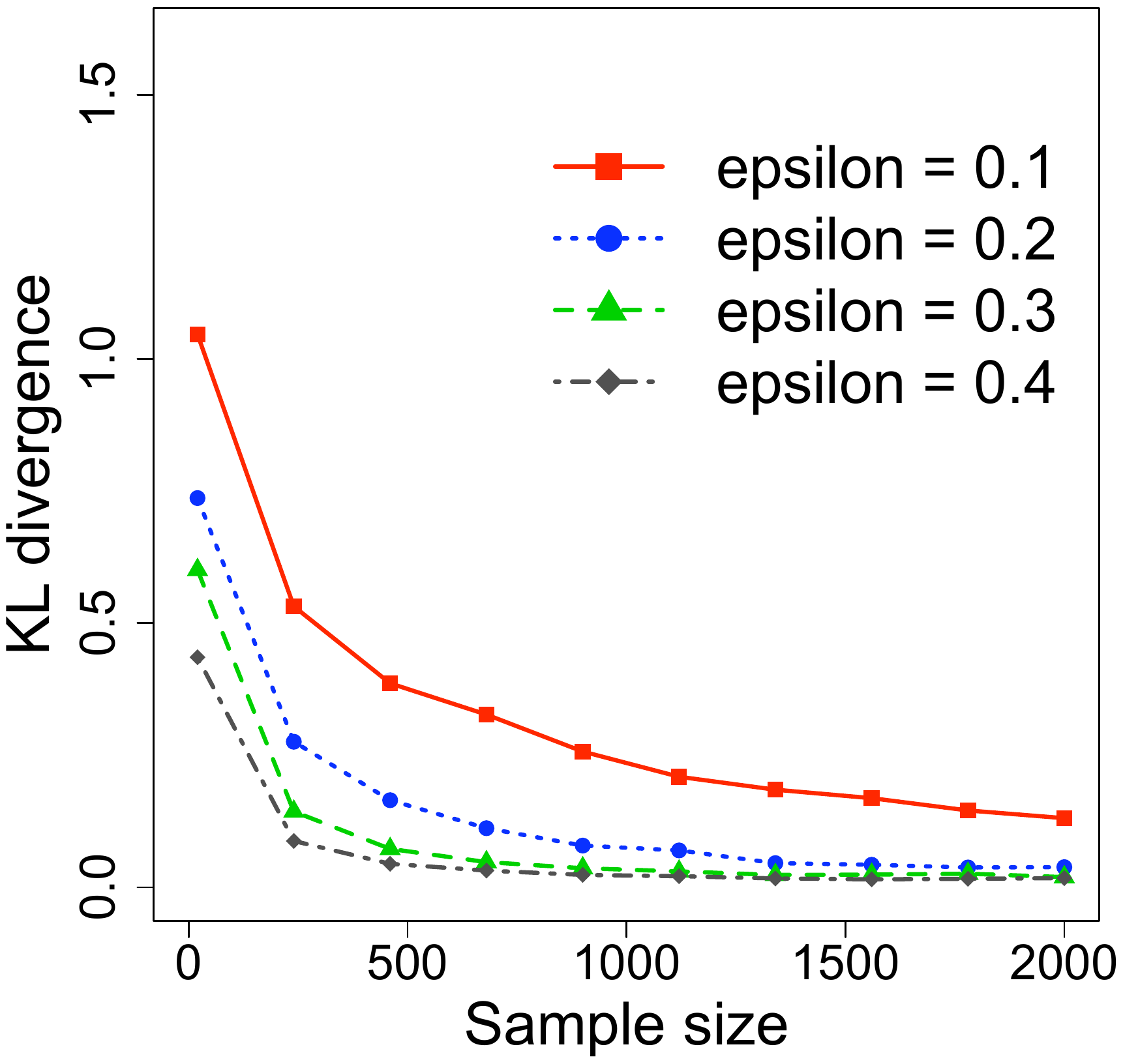}
\includegraphics[scale=0.23]{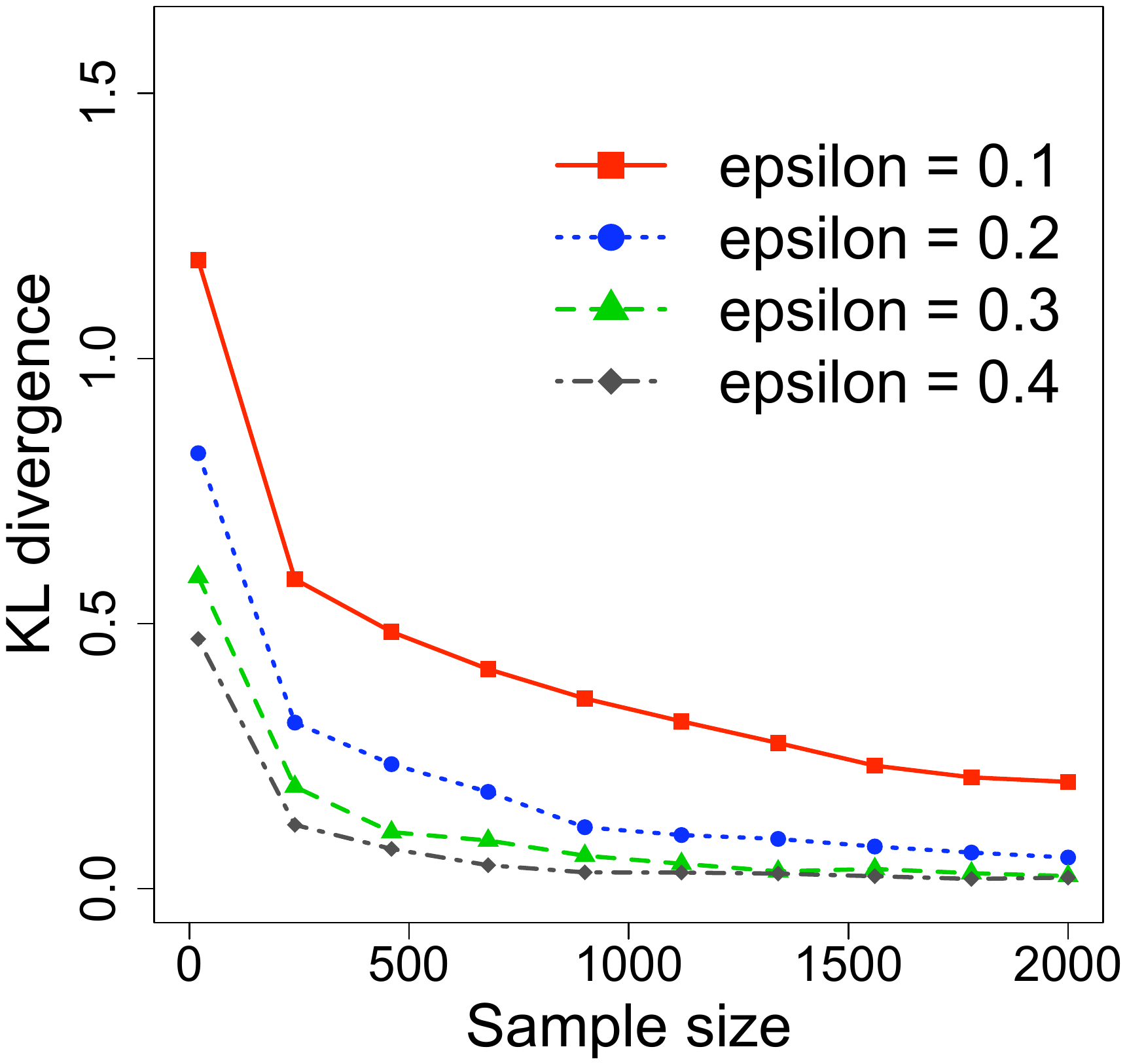}
\includegraphics[scale=0.23]{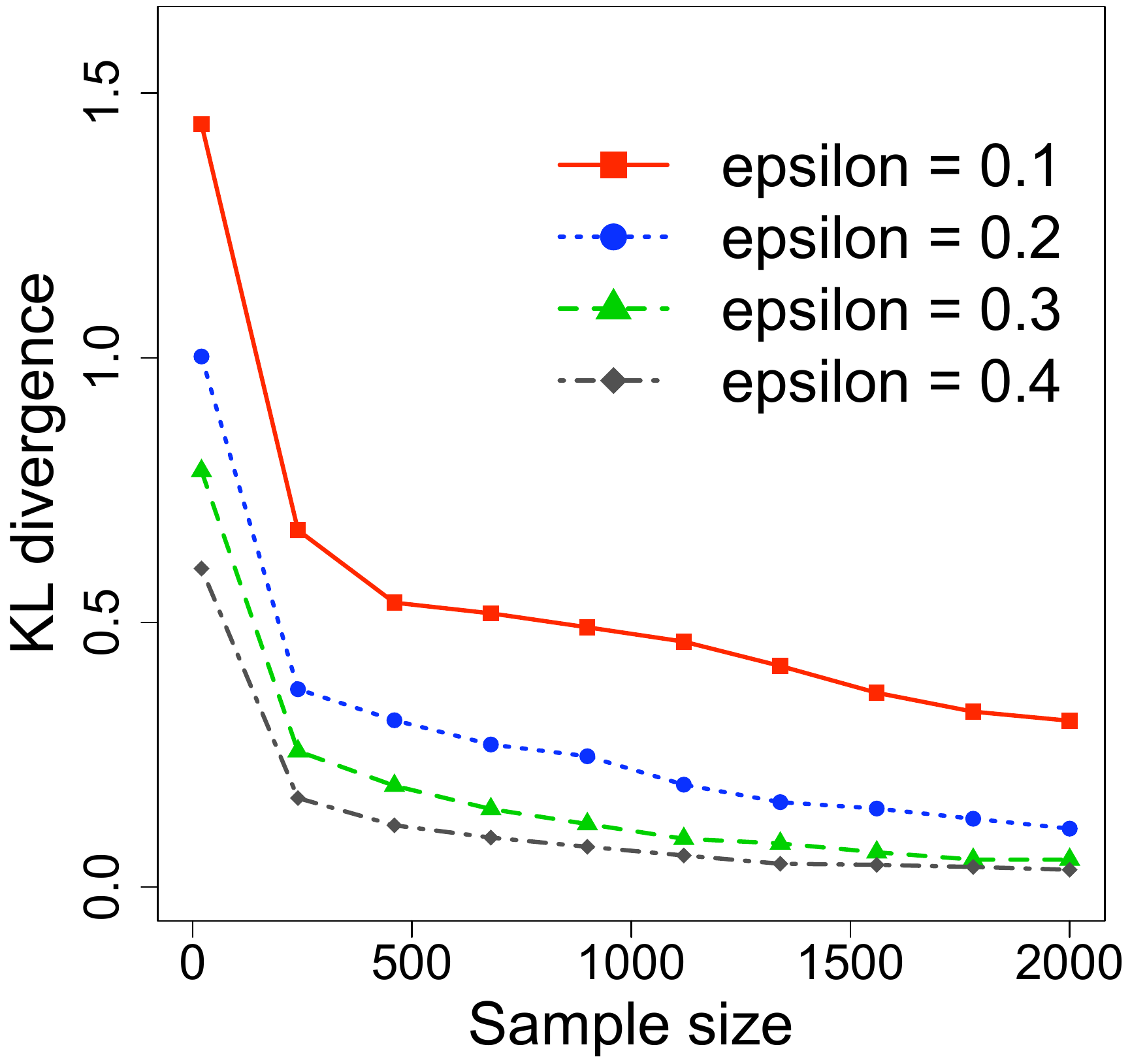}
\includegraphics[scale=0.23]{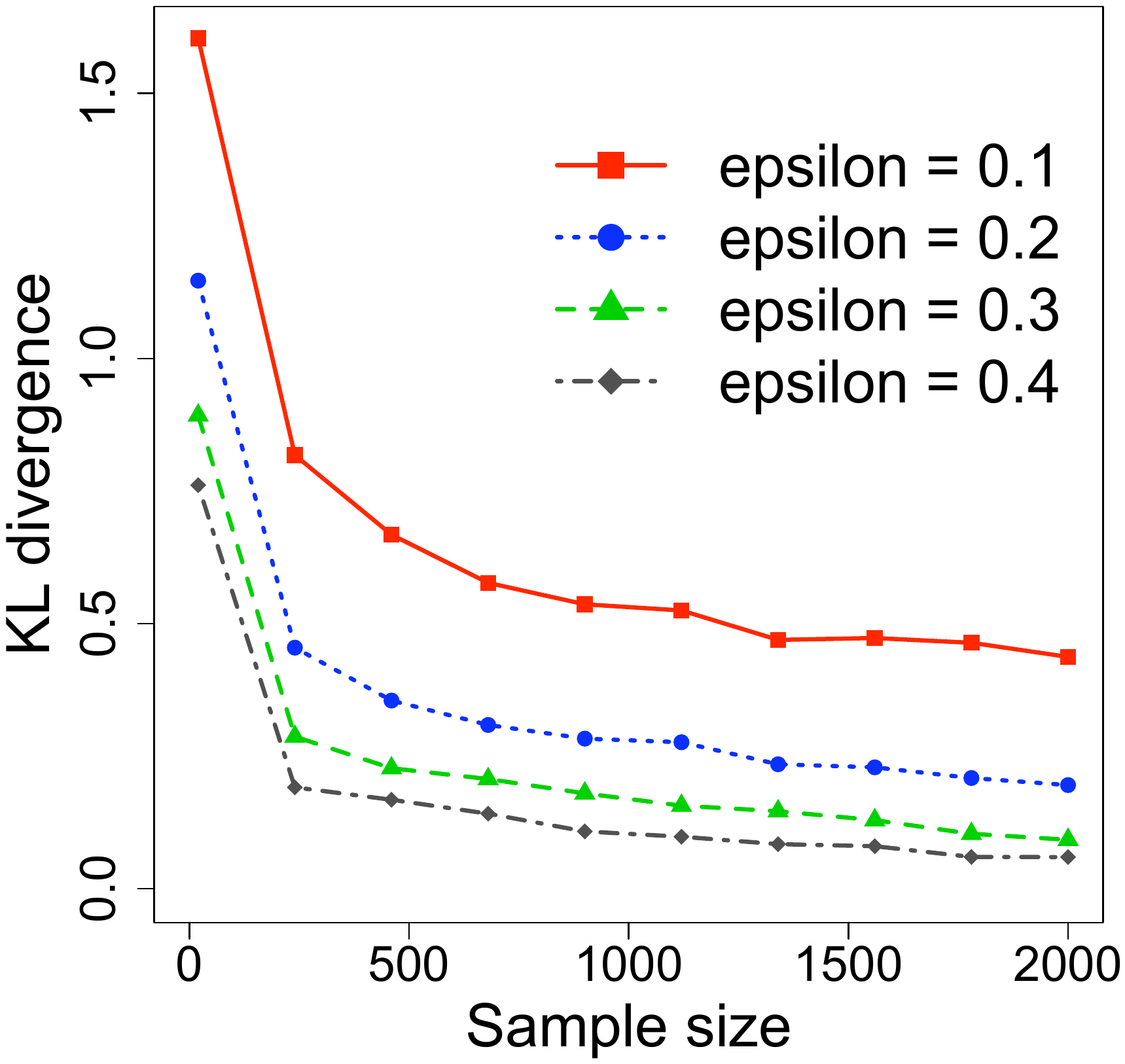}
\caption{KL divergence between the original $\chi^2$-statistic and the private $\chi^2$-statistic based on the frequency table (a) left, (b) middle left, (c) middle right, and (d) right.}
\label{fig_chi}
\end{figure}

%[ADDITIONAL DISCUSSION] 

\begin{figure}[!b]
\centering
\includegraphics[scale=0.23]{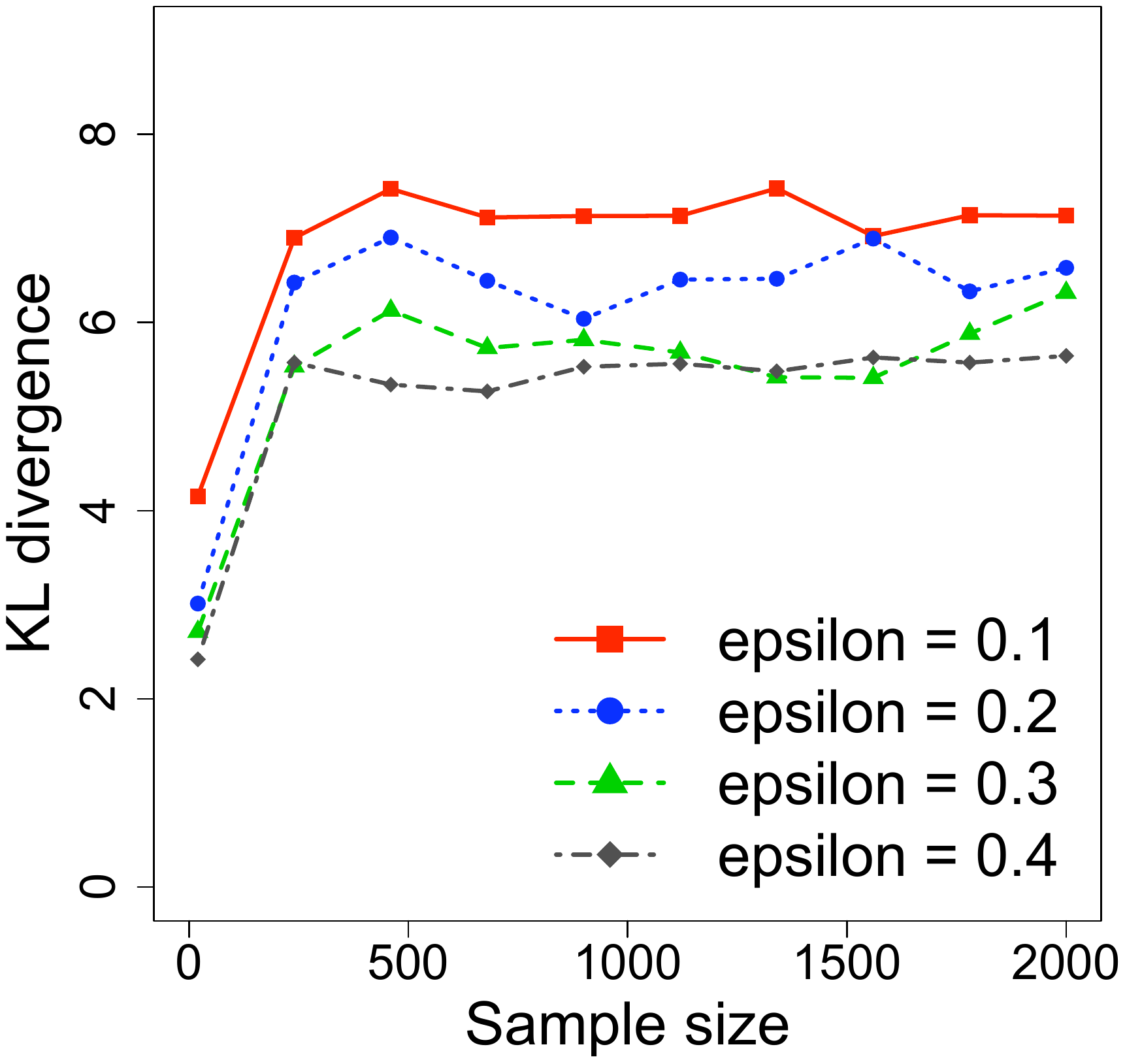}
\includegraphics[scale=0.23]{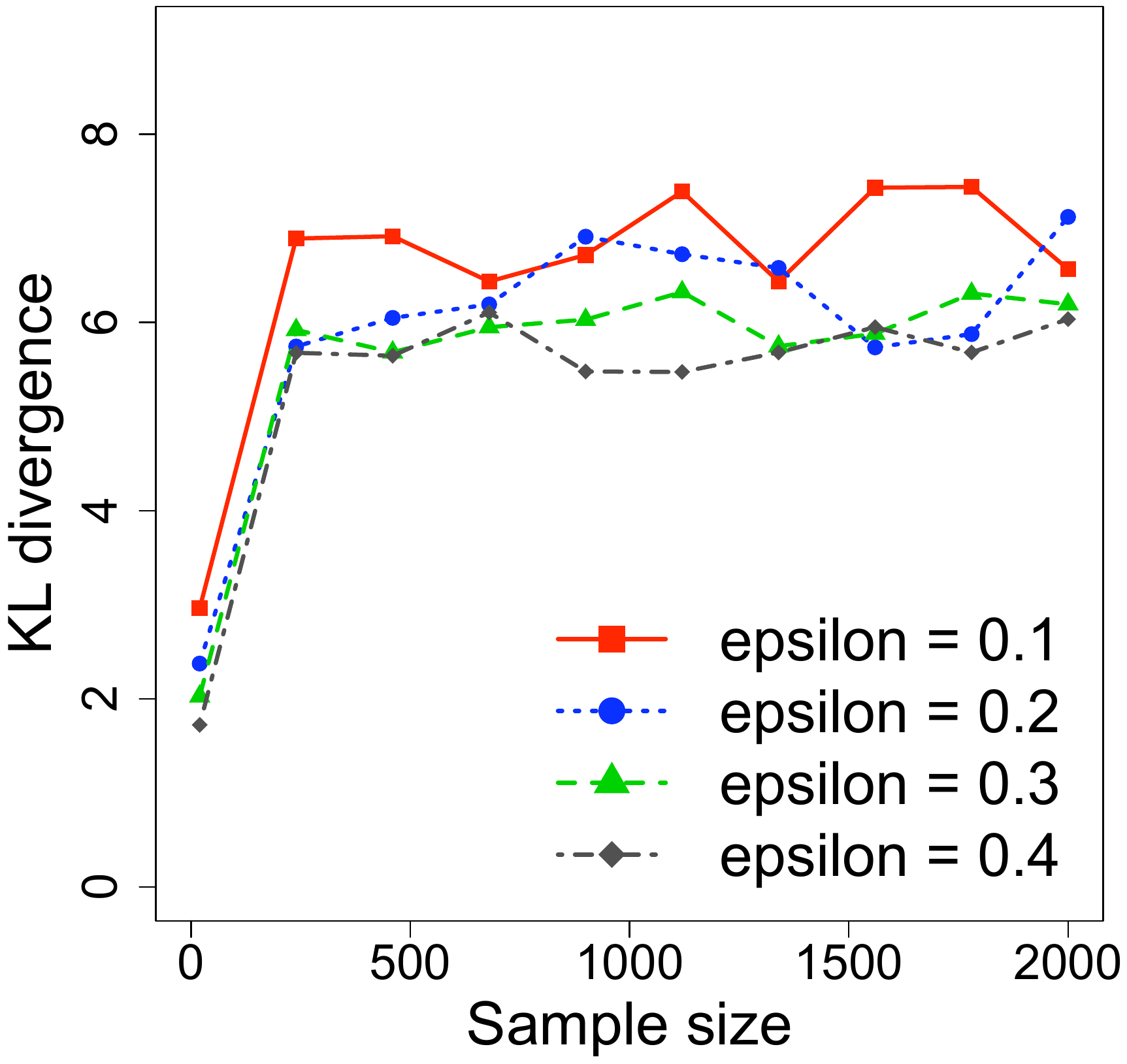}
\includegraphics[scale=0.23]{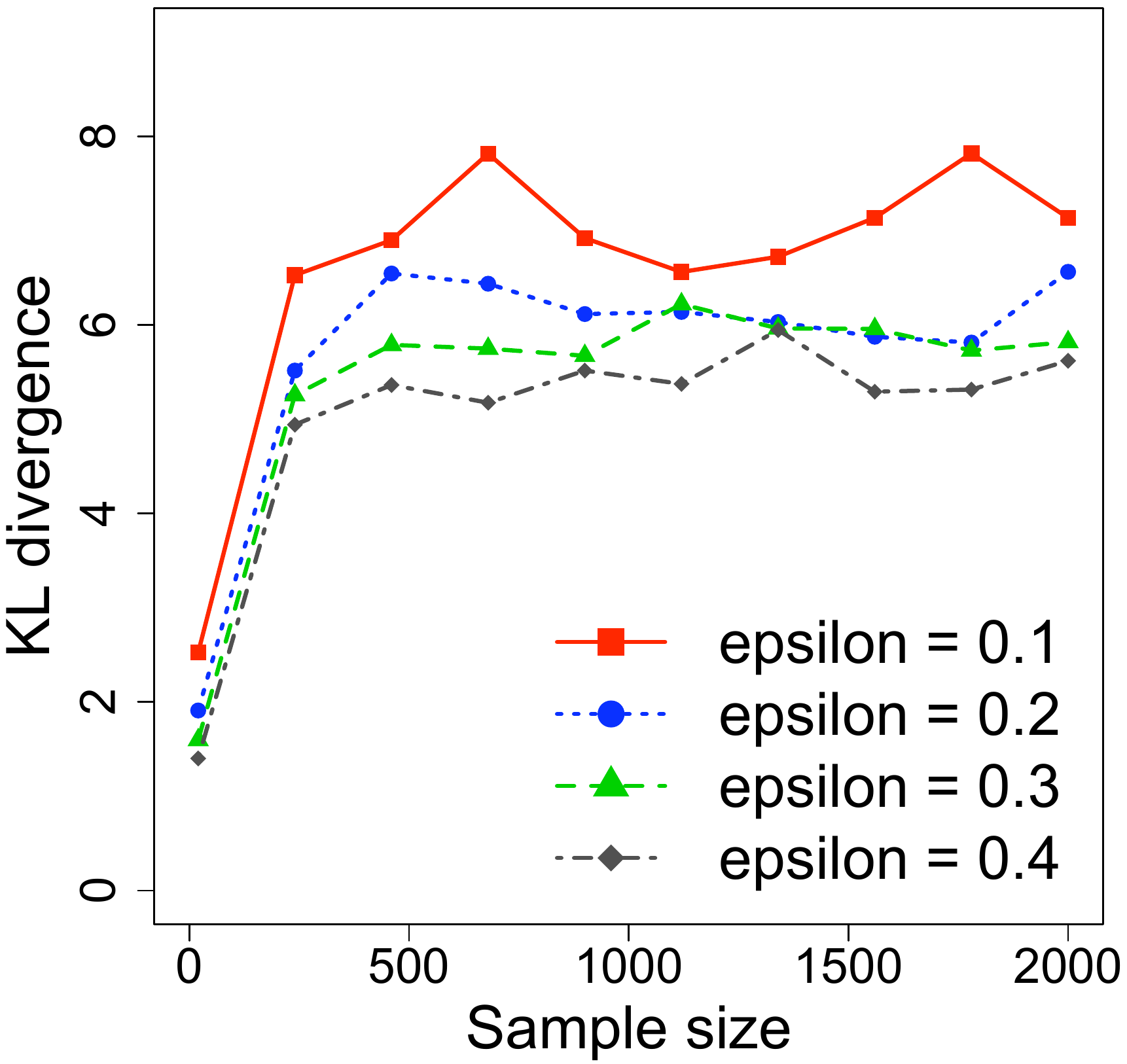}
\includegraphics[scale=0.23]{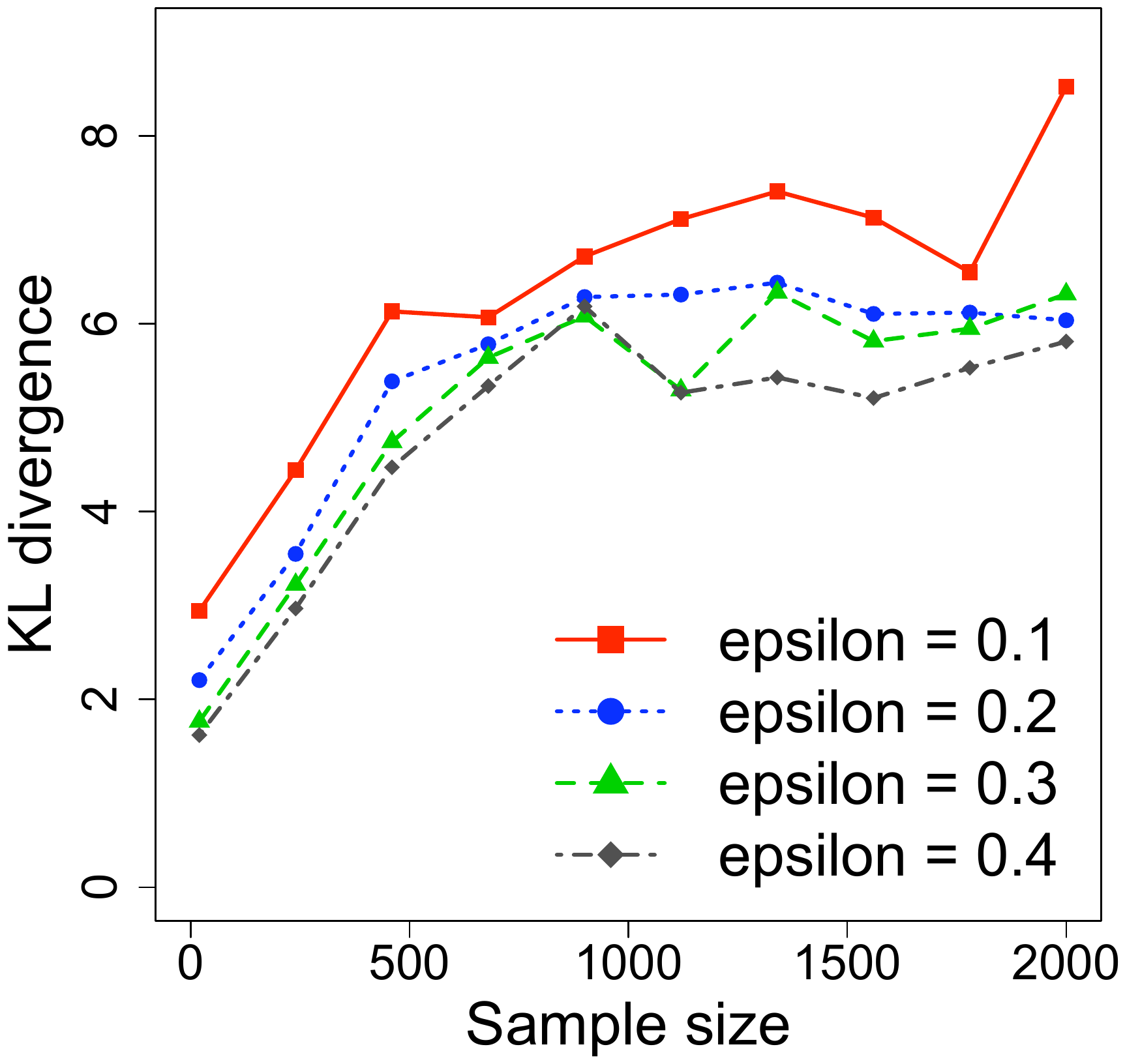}
\caption{KL divergence between the original $p$-values and the private $p$-values based on the frequency table (a) left, (b) middle left, (c) middle right, and (d) right.}
\label{fig_pvalue}
\end{figure}

\subsection{Differentially-Private $p$-Values}

We did a similar analysis on the $p$-values following the proposed release mechanism of adding Laplace noise according to Theorem~\ref{thm_p}. Based on the frequency tables in (\ref{eq_frequencies}), we computed the KL divergence between the original and private $p$-values for increasing $N$ and for four different privacy levels. The resulting plots are shown in Figure \ref{fig_pvalue}. Similarly to the $\chi^2$-statistics, the smaller $\epsilon$, the larger the KL divergence is. However, the relation between the KL divergence and the number of individuals, resp.~the original $\chi^2$-statistic, is reversed since,   for the $\chi^2$-distribution with 2 degrees of freedom, the $\chi^2$-statistic is proportional to the logarithm of the $p$-value. The larger the $\chi^2$-statistic, the smaller the $p$-value and hence the smaller the signal to noise ratio. The jumps in the figures arise because we project the perturbed $p$-values which fall outside the interval $[0,1]$ to $0$ or $1$, respectively. Although there is a one-to-one correspondence between the $\chi^2$-statistics and the $p$-values, the $\chi^2$-statistics have a much smaller KL divergence and are therefore better suited for privacy purposes.

%[WHAT ABOUT ADDING NOISE TO log p-value???]

Projecting the $p$-values onto a region of interest as described in Corollary \ref{cor} results in  plots similar to those in Figure \ref{fig_pvalue}; %, primarily  because 
the plots  depend on how much smaller the $p$-value under consideration is compared to 1 in the case of Theorem \ref{thm_p} and $p^*$ in the case of Corollary \ref{cor}.

%HERE response to \section{Comment (a)}
%\emph{Expanding the discussion of adding noise directly to the p-value. The noisy p-value becomes even more of a random variable (a distribution due to sampling and perturbation). If the researchers originally wanted a p-value cutoff of 0.05 on the original data, what cutoff should they use for the perturbed p-value to control their Type I \& II errors?}
%\vspace{0.3cm}

Our analysis and the plots in Figure \ref{fig_pvalue} strongly suggest that perturbing the p-values to achieve $\epsilon$-differential privacy leads to too much noise. Making inference based on such perturbed p-values seems impossible. However, it is a valid question to ask whether there might exist a cut-off which could control the Type I \& Type II errors.

We analyze this question by sampling 500 true positives ($p$-values in $[0,0.05]$) and 500 true negatives ($p$-values in $[0.05,1]$) uniformly and adding Laplace noise with scale $\exp(-\frac{2}{3})/\epsilon$. We represent the simulated data in an ROC plot, where we report for all possible cut-off values the resulting Type I and Type II errors. These plots for four levels of privacy, namely $\epsilon=0.1, 0.2, 0.3, 0.4$ are shown in Figure \ref{fig:ROC_pvalues}. We especially indicate the point corresponding to the usual cut-off of 0.05.

Figure \ref{fig:ROC_pvalues} confirms that using the perturbed $p$-values as a test for independence is not much better than a random test, independent of the chosen cut-off. Choosing a cut-off of 0.05 seems reasonable, but it is anyways impossible to control the Type I \& Type II errors. An interesting feature in the plots are the long straight lines going from both corners along the diagonal. These lines arise since we project the perturbed $p$-values which fall outside the interval $[0,1]$ to either 0 or 1. These plots show again that the perturbed p-values are dominated by these projected 0's and 1's rendering a test based on the perturbed $p$-values uninformative.

\begin{figure}[t]
\centering
\subfigure[$\epsilon=0.1$]{\includegraphics[scale=0.22]{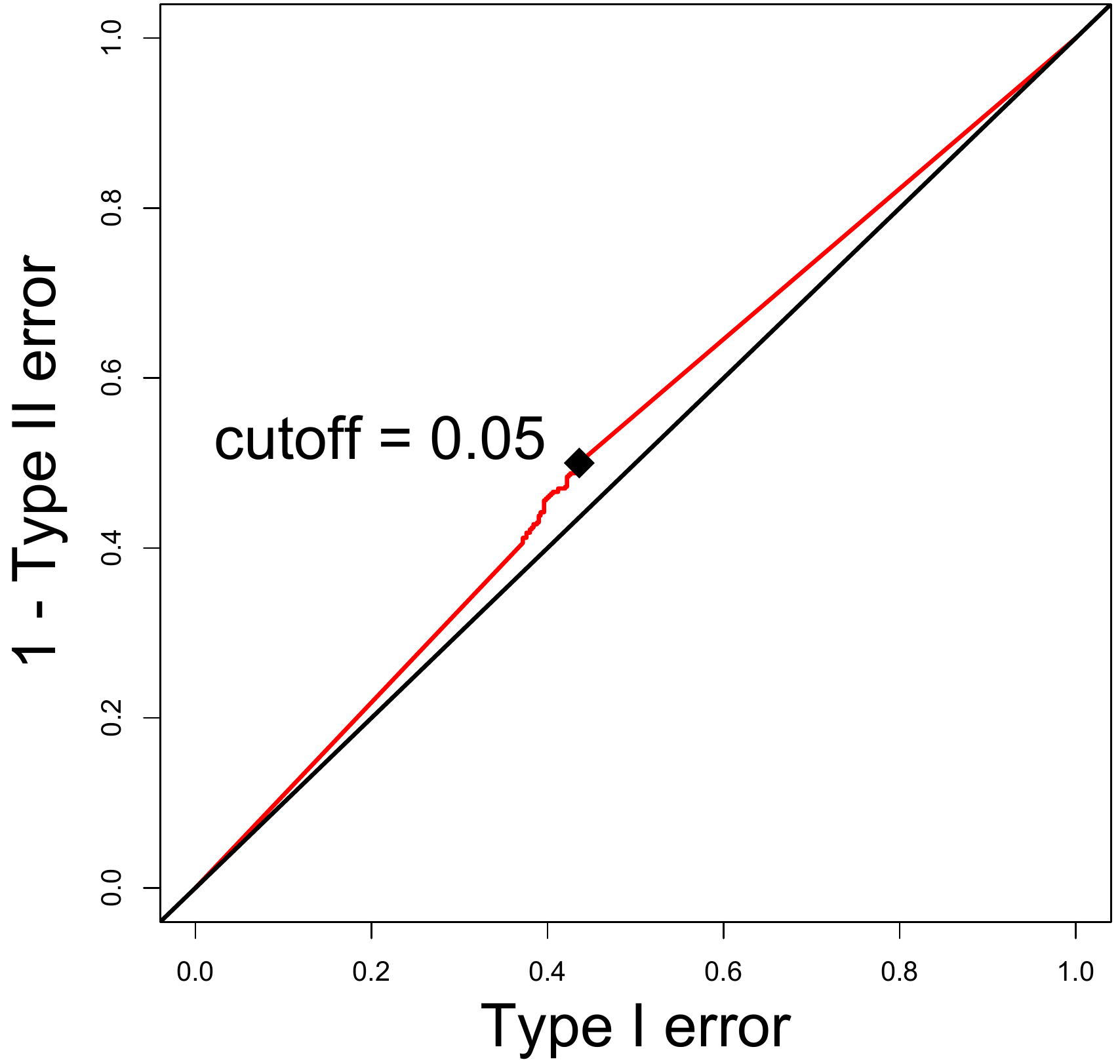}\label{fig:ROC_pvalues:1}} \;
\subfigure[$\epsilon=0.2$]{\includegraphics[scale=0.22]{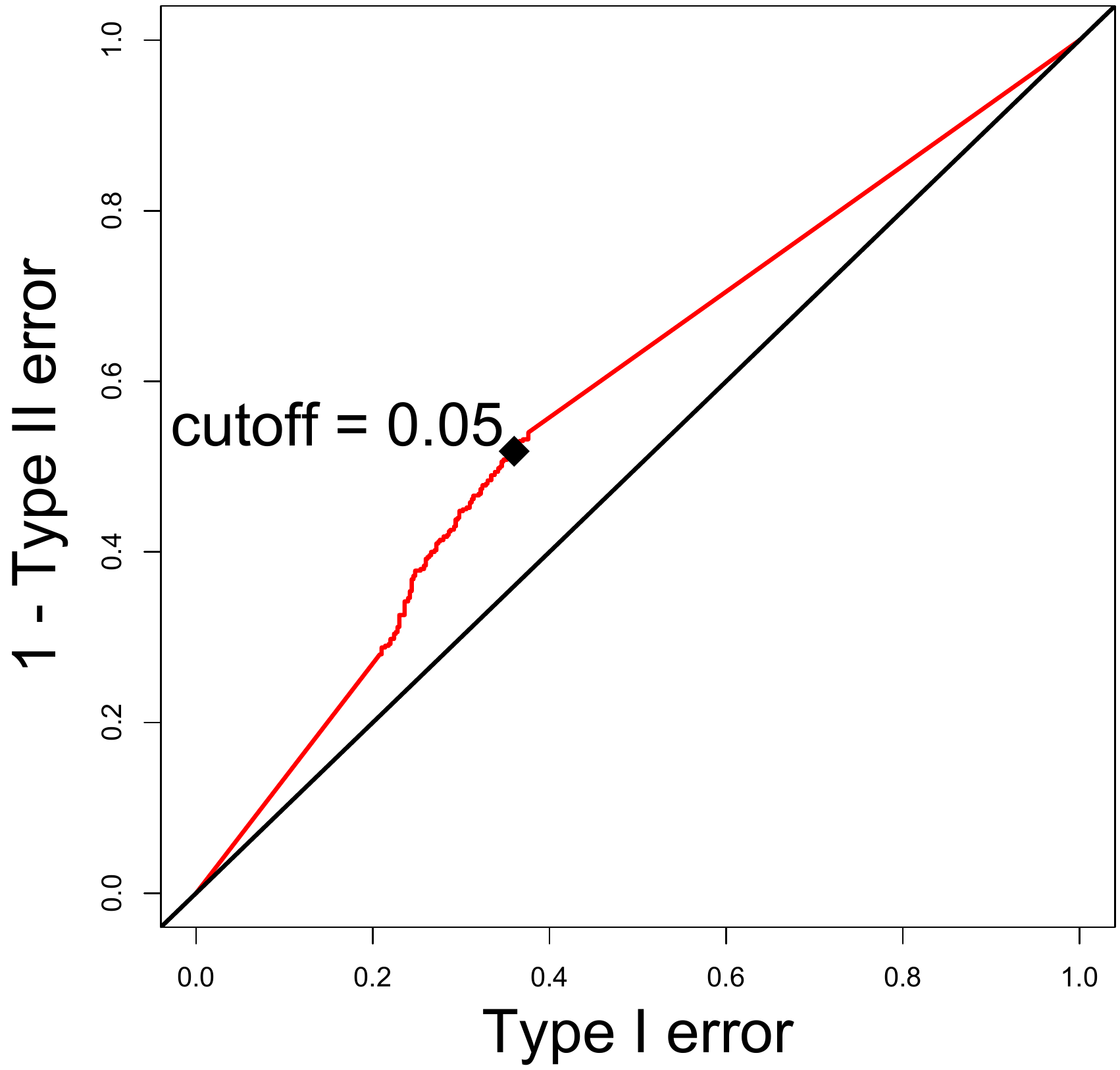}\label{fig:ROC_pvalues:2}} \;
\subfigure[$\epsilon=0.3$]{\includegraphics[scale=0.22]{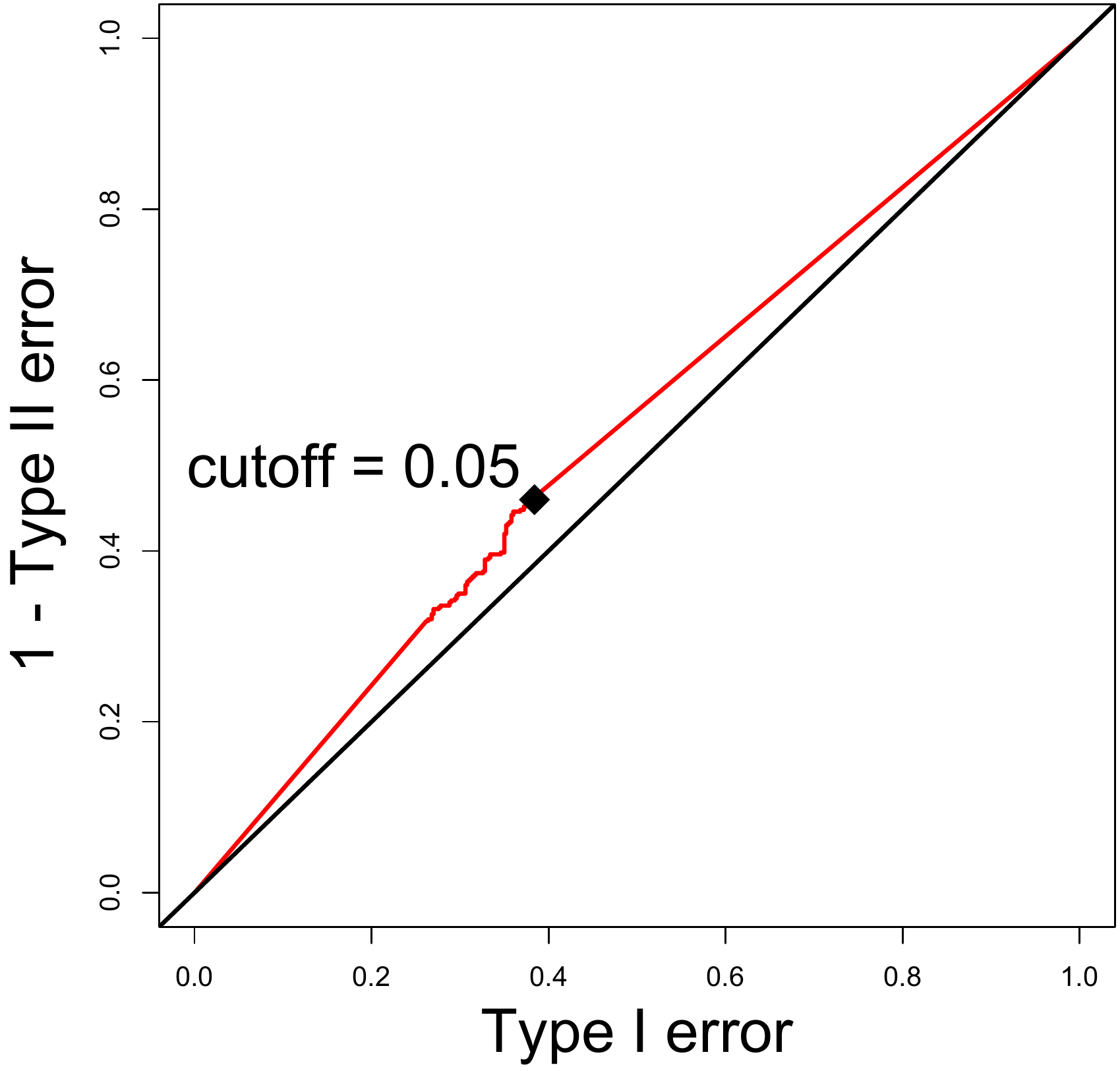}\label{fig:ROC_pvalues:3}} \;
\subfigure[$\epsilon=0.4$]{\includegraphics[scale=0.22]{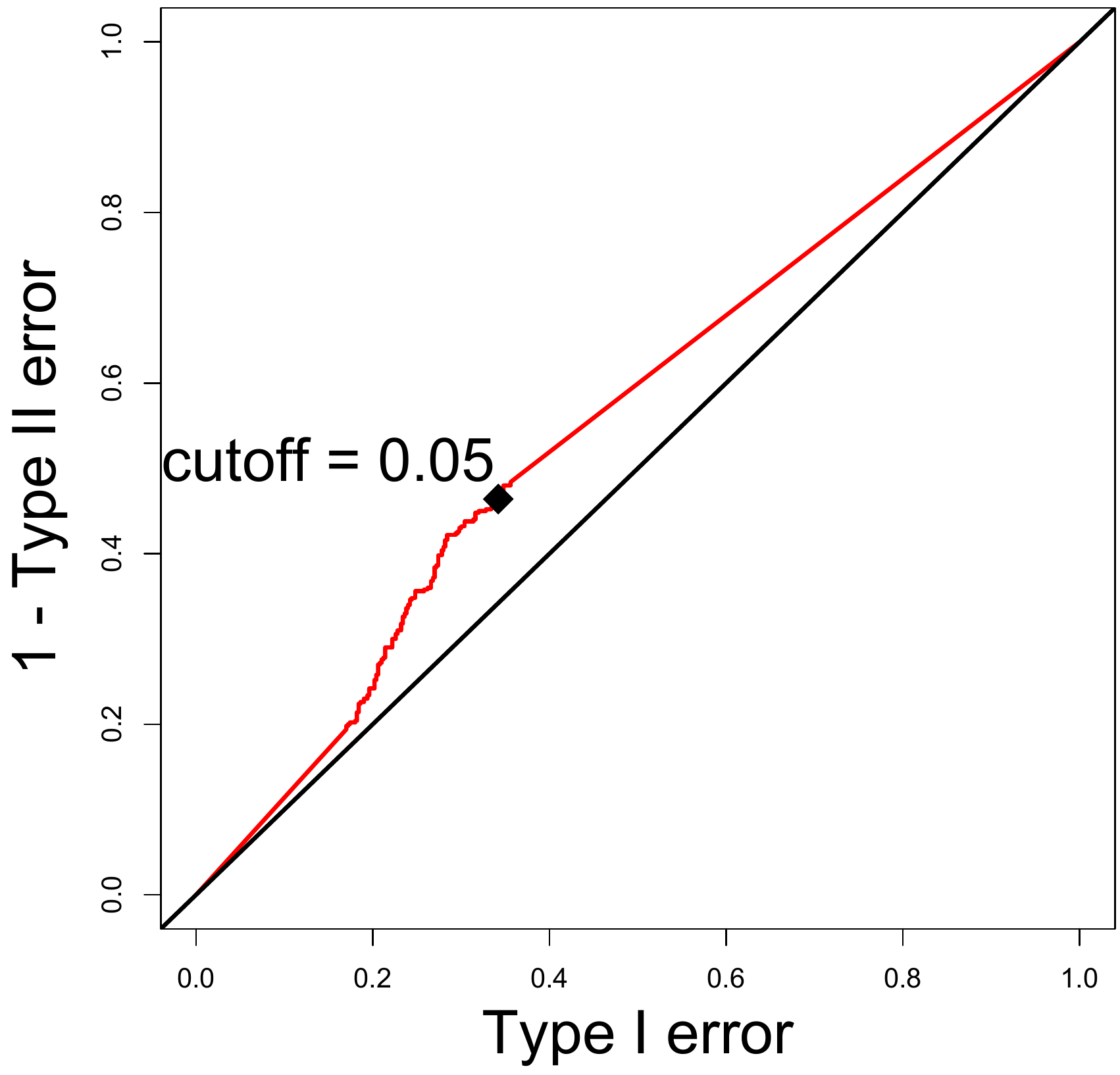}\label{fig:ROC_pvalues:4}}
\caption{ROC curves for the perturbed p-values.}
\label{fig:ROC_pvalues}
\end{figure}
%%% end response to comment(a) %%%%

\subsection{Releasing the $M$ Most Relevant SNPs with Respect to a Specific Phenotype}

Practitioners are often interested in finding and releasing the most relevant (i.e., most statistically and practically significant) SNPs. Here we analyze what sample size $N$ is needed in order to recover the two causative SNPs in the HAP-SAMPLE simulations from the private $\chi^2$-statistics. We chose $M=3$ and plotted the frequencies (based on 1,000 private $\chi^2$-statistics) for which one or both of the two causative SNPs were among the three highest ranked private $\chi^2$-statistics computed according to Algorithm~\ref{alg_chi_p}. We performed this analysis for increasing sample size $N$ and for four different privacy levels. We used the simulated HAP-SAMPLE data consisting of around 10,000 SNPs total with two causative SNPs under the additive model with MAF=0.25 and MAF=0.4. The resulting bar charts are shown in Figure \ref{fig_bar_graphs}.

\begin{figure}[!b]
\centering
\includegraphics[scale=0.45]{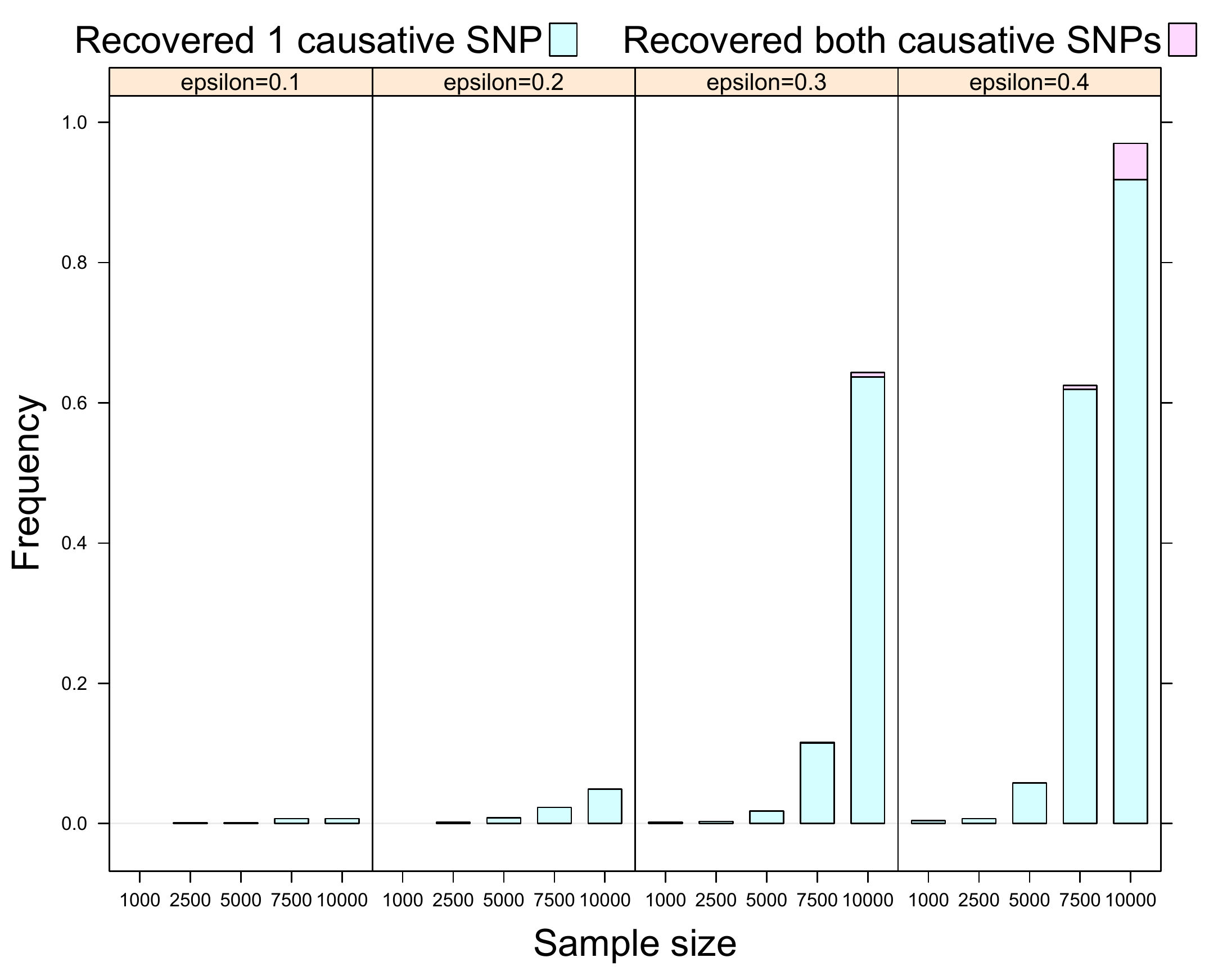}
\includegraphics[scale=0.45]{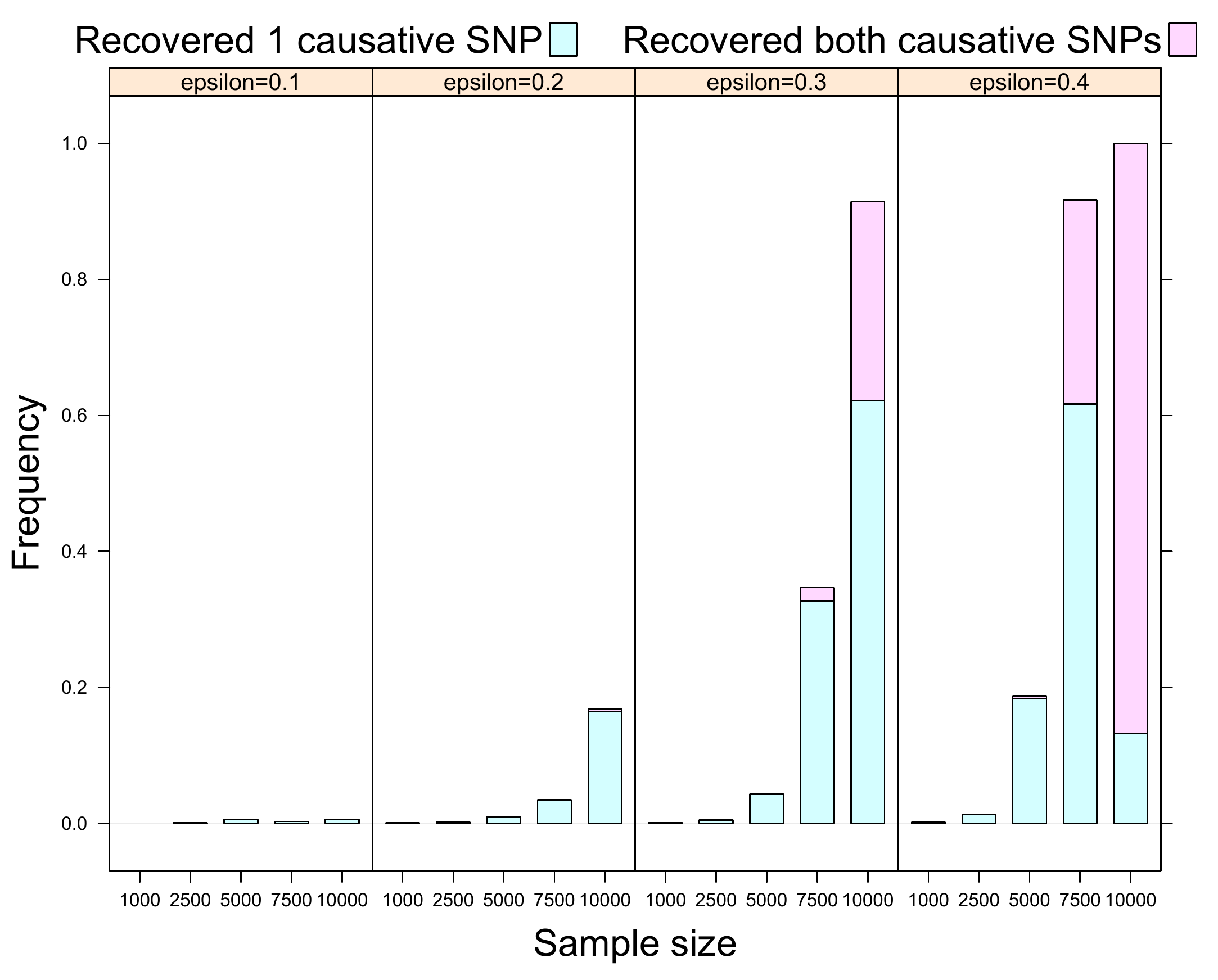}
\caption{Bar charts representing the frequencies for which one  or both of the two causative SNPs were among the three highest ranked private $\chi^2$-statistics under the additive model with MAF=0.25 (top) and MAF=0.4 (bottom).}
\label{fig_bar_graphs}
\end{figure}

As we expect, a larger value of $\epsilon$ (i.e., less noise/less privacy) results in a higher chance of releasing the truly causative SNPs. We also observe that the smaller the MAF, the more data we need to detect the causative SNPs at a fixed level of $\epsilon$. For example, for $\epsilon=0.4$, Figure \ref{fig_bar_graphs} shows that for MAF=0.4 we need about 7,500 individuals to detect the causative SNPs whereas for MAF=0.25 we need about 10,000 individuals. A smaller MAF corresponds to a sparser table, and we are in a similar situation to that described in \cite{Fienberg_2010}, where it is shown that for sparse tables differential privacy requires adding a lot of noise, often with the result of impairing statistical inference. Our results support the traditional trade-off: in order to detect important effects, we need to either relax the privacy constraint or increase the total number of individuals massively. 

An alternative to adding noise to the data we want to release is to add noise to the analysis itself. We explain this approach for  GWAS in the following section.

\section{Extended Work: Differentially-Private Algorithm for Detecting Epistasis}

As we just saw, the sparseness of GWAS data requires an unrealistically large number of individuals in each study or a relaxation of the privacy level. In order to deal with sparseness, methods have been proposed, where the Laplace noise is added to the analysis directly instead of to the output. Another advantage of such an approach is that it allows the analysis of models that integrate information across SNPs. Here we present an $\epsilon$-differentially logistic regression approach that is directly applicable to GWAS.

Most methods for detecting epistasis are based on a two-stage approach. First, all SNPs are filtered e.g.~using $\chi^2$-statistics or $p$-values, to reduce the potential interacting SNPs to a small number. The loci achieving some threshold are then further examined for interactions. A widely used  test for detecting gene-gene interactions on a small number of SNPs is a penalized logistic regression, e.g.~the $L_2$-regularized logistic regression proposed by Park and Hastie \cite{Park}. By adapting the work of Bhaskar et al.~\cite{Bhaskar} and Chaudhuri et al.~\cite{Chaudhuri}, we derive a privacy-preserving method for detecting epistasis, where both stages in the two-stage approach satisfy differential privacy.

We use the first two steps in Algorithm \ref{alg_chi_p} to chose a subset of interesting SNPs of size $M$ in a differentially private way. Park and Hastie~\cite{Park} suggest an $L_2$-regularized logistic regression in order to detect epistasis within a small subset of SNPs.  Chaudhuri et al.  \cite{Chaudhuri} demonstrated how to perturb the objective function for privacy-preserving machine-learning algorithm designs if the loss function and the regularizer satisfy certain convexity and differentiability criteria. In the following, we outline how to apply their objective perturbation in order to find a differentially private algorithm for detecting epistasis.

Let $y=(y_1,\dots , y_N)$ denote the disease status of the N individuals. Note that in this section we encode the diseased status by 1 and the non-diseased status by -1. Let $x_i\in \mathbb{R}^{p+1}$ denote the feature vector for the $i^{\textrm{th}}$ individual. The first entry corresponds to the intercept. The encoding of the features is explained via an example. We will look at a model with two SNPs including their interaction. SNP1 takes the three states 0, 1, and 2, which are encoded by 100, 010, and 001. Similarly for SNP2. The interaction term SNP1$\times$SNP2 takes the states 00, 01, 02, 10, 11, 12, 20, 21, 22 and is encoded by $100000000, 010000000, \dots , 000000001$. So an individual with genotype 12, who is not diseased would have
$$x=(1,0,1,0,0,0,1,0,0,0,0,0,1,0,0,0),\qquad y=-1.$$
Let $K-1$ be the total number of effects in the model (including main and higher-order effects). It is important to note that $|\!|x_i|\!|_2\leq K$. 

The objective function described in Park and Hastie~\cite{Park} is
$$L(\beta)=\sum_{i=1}^{N}\log(1+\exp(-y_i\beta^Tx_i))+\frac{1}{2}\beta^T\Lambda\beta,$$
where $\Lambda$ is of the form $(0,\lambda,\dots ,\lambda)$, i.e.~$\beta_0$ is not penalized.  They use the Newton-Raphson method  for the optimization and forward selection and backward deletion steps based on an Akaike Information Criterion (AIC) or Bayesian Information Criterion (BIC) score to select model size and important factors.

We can apply the approach of Chaudhuri et al. \cite{Chaudhuri} to perturb the objective function such that the algorithm satisfies $\epsilon$-differential privacy. We are interested in the following perturbed objective function:
$$L_{\textrm{priv}}(\beta)=\sum_{i=1}^{N}\log(1+\exp(-y_i\beta^Tx_i))+\frac{1}{2}\beta^T\Lambda\beta+\frac{1}{N}b^T\beta,$$
where b is noise drawn from a distribution with density
$$f(b)=\frac{1}{\alpha}\exp(-k|\!|b|\!|_2)$$
and $k$ is a constant and $\alpha$ the normalizing constant. 

Following the proposal by Park and Hastie~\cite{Park} we make use of forward selection and backward deletion steps based on an AIC or BIC score to select model size; however, we replace the optimization step in their method by Algorithm \ref{alg_log_reg}.

%\begin{alg}
\begin{algorithm}[!b]
\caption{$\epsilon$-Differentially Private Algorithm for Detecting Epistasis}
\begin{algorithmic}
\begin{STATE}
\label{alg_log_reg}

{\bf Input:}
The data vectors $x_i, y_i$, where $i=1,\dots ,N$ and parameters $\epsilon$, $\lambda$, and $c$. 

{\bf Output:}
The output consists of the noisy effects.

\begin{enumerate}
\item[1.] Let $\epsilon'=\epsilon-\log(1+\frac{2cK}{N\lambda}+\frac{c^2K^2}{N^2\lambda^2})$. If $\epsilon'>0$, then $\delta=0$, else $\delta=\frac{cK}{N(e^{\epsilon/4}-1)}-\lambda$ and $\epsilon'=\epsilon/2K$.
\item[2.] Draw $b$ from a distribution with density $f(b)=\frac{1}{\alpha}\exp(-\frac{\epsilon|\!|b|\!|_2}{2})$.
\item[3.] Compute $\beta_{\textrm{priv}}=\textrm{argmin} (L_{\textrm{priv}}(\beta) +\frac{1}{2}\delta|\!|\beta|\!|_2)$.
\end{enumerate}
%\end{alg}

\end{STATE}
\end{algorithmic}
\label{alg_log_reg}
\end{algorithm}
%\vspace{-0.3cm}

\begin{thm}
Algorithm \ref{alg_log_reg} is $\epsilon$-differentially private. 
\end{thm}

\begin{proof}
The proof follows from Theorem 9 in \cite{Chaudhuri}, and by taking into account  the fact that $|\!|x_i|\!|_2\leq K$ for our application.
\end{proof}

This result allows us to move away from a SNP-by-SNP analysis to an integrated approach without relaxing privacy. Applying this method to  actual GWAS data is part of ongoing work.  

\section{Conclusion}
In this paper, we have demonstrated that it is possible, using the  formal privacy guarantees of differential privacy, for NIH and other  GWAS data repositories as well as ``GWAS data owners" to release at least \mbox{some genetic} data required by practitioners.  %We described an $\epsilon$-differentially private methodology for GWAS data sharing. 
More specifically, we described a privacy-preserving release of aggregate minor allele frequencies and the release of differentially-private $\chi^2$-statistics and $p$-values.  We also provided a differentially private algorithm for releasing these statistics for the most relevant SNPs. 

Our simulations, however, indicate that for bigger and sparse data the release of simple summary statistics is problematic and not sufficient from both privacy and utility perspectives.  The release of summary statistics may be at least in part sufficient for traditional piecewise SNP-by-SNP analysis. More specifically, our results on finite sample properties of differentially-private $\chi^2$-statistics  show that adding noise directly to the $\chi^2$-statistic achieves the best trade-off between privacy and utility in comparison to adding noise to the $p$-values or cell entries themselves, in particular for tables with small to moderate counts and overall samples size. %The next step would be to consider a more careful evaluation of performance of $p$-values based on derived convolution. 
However, we require more complex methodology  to deal with more sparse data and models that integrate across SNPs to detect epistasis. To address this problem, we outlined an $\epsilon$-differentially private algorithm for a specific form of penalized logistic regression. 
This is but one of the newer methods being introduced into the statistical literature for GWAS, but we expect that the general strategy suggested here might be adaptable for other statistical methods, e.g., for sparse partitioning \cite{speed2011sparse}.

%, and other variations on differential privacy, e.g., exponential mechanism \cite{}. 
Since the introduction of differential privacy by \cite{Dwork}, and in particular $\epsilon$-differential privacy, many additional variations along with their considerations with respect to statistical analysis have been proposed (e.g.,  more recently \cite{Hardt:2010fk}). To further improve the privacy-utility tradeoffs for GWAS, the future research would consider such alternate mechanisms. 

%such as $(\epsilon, \delta)$-differential privacy, exponential mechanism, and their considerations with respect to statistical analysis (e.g., a few relevanat REFS see REFs).  

%[WHAT DOES THIS MEAN FOR THE GWAS and real data setting???] [WHAT DO WE LEARN FROM SIMULATIONS?] [WHAT's the NEXT STEP?] 

%DISCUSSION ON \em{sensitivity}. We used L1 norm as a measure? Can we do better? 

%DISCUSSION what's hot in the stats literature!!! sparse partitioning. 

\section*{Acknowledgment}
This research was supported in part by NSF Grants BCS-0941553 and BCS-0941518
 to Pennsylvania State University and Carnegie Mellon University, respectively.

%The research reported
%here was supported in part by NSF Grants BCS-0941553 and BCS-0941518
%to the Department of Statistics at the Pennsylvania State University and at Carnegie Mellon University, respectively. 



%\bibliographystyle{IEEEtran}

\bibliographystyle{abbrv}
{\bibliography{IEEEabrv,references_GWAS}
}

\end{document}